\documentclass[submission,copyright,creativecommons]{eptcs}
\providecommand{\event}{TARK 2025} 
\usepackage{iftex}

\usepackage{latexsym}
\usepackage{amssymb}
\usepackage{amsmath}
\usepackage[all, knot]{xy}
\usepackage{graphicx}
\usepackage{hyperref}
\usepackage{multirow}
\usepackage{tikz,pgf} 
\usetikzlibrary{backgrounds}
\usetikzlibrary{arrows,positioning,graphs}

\newcommand{\M}{\mathcal{M}}
\newcommand{\N}{\mathcal{N}}
\newcommand{\SC}{\mathcal{C}}
\newcommand{\SD}{\mathcal{D}}
\newcommand{\F}{\mathcal{F}}

\newcommand{\U}{\mathcal{U}}

\newcommand{\lr}[1]{\langle #1 \rangle}

\newcommand{\LL}{\mathcal{L}}
\newcommand{\K}{\mathsf{K}}
\newcommand{\DK}{\widehat{\mathsf{K}}}

\newcommand{\bis}{\mathrel{\mathchoice%
{\raisebox{.3ex}{$\,
  \underline{\makebox[.7em]{$\leftrightarrow$}}\,$}}%
{\raisebox{.3ex}{$\,
  \underline{\makebox[.7em]{$\leftrightarrow$}}\,$}}%
{\raisebox{.2ex}{$\,
  \underline{\makebox[.5em]{\scriptsize$\leftrightarrow$}}\,$}}%
{\raisebox{.2ex}{$\,
  \underline{\makebox[.5em]{\scriptsize$\leftrightarrow$}}\,$}}}}

\makeatletter
\newcommand{\shorteq}{
  \settowidth{\@tempdima}{-}
  \resizebox{\@tempdima}{\height}{=}
}
\makeatother

\newcommand{\assn}{\!:\!\!\shorteq\,}
\newcommand{\sassn}{{[:\shorteq]}}

\newtheorem{theorem}{Theorem}[section]

\newtheorem{lemma}[theorem]{Lemma}
\newtheorem{definition}[theorem]{Definition}
\newtheorem{proposition}[theorem]{Proposition}

\newtheorem{corollary}[theorem]{Corollary}

\newenvironment{proof} {\textsc{Proof.}\quad} {\hfill $\Box$\\}

\newenvironment{proofcomp} {\textsc{Proof of Theorem \ref{thm.comp}.}\quad} {\hfill $\Box$\\}
\newenvironment{proofanf} {\textsc{Proof of Proposition \ref{prop.anf}.}\quad} {\hfill $\Box$}
\ifpdf
  \usepackage{underscore}        
  \usepackage[T1]{fontenc}      
\else
  \usepackage{breakurl}         
\fi

\title{Impure Simplicial Complex and Term-Modal Logic with Assignment Operators}
\author{Yuanzhe Yang
\institute{Department of Philosophy and Religious Studies\\
Peking University\\
Beijing, China}
\email{1900014924@pku.edu.cn}
}
\def\titlerunning{Impure Simplicial Complex and Term-Modal Logic with Assignment Operators}
\def\authorrunning{Y. Yang}

\begin{document}
\maketitle

\begin{abstract}
Impure simplicial complexes are a powerful tool to model multi-agent epistemic situations where agents may die, but it is difficult to define a satisfactory semantics for the ordinary propositional modal language on such models, since many conceptually dubious expressions involving dead agents can be expressed in this language. In this paper, we introduce a term-modal language with assignment operators, in which such conceptually dubious expressions are syntactically excluded. We define both simplicial semantics and first-order Kripke semantics for this language, characterize their respective expressivity through notions of bisimulation, and show that the two semantics are equivalent when we consider a special class of first-order Kripke models called local epistemic models. We also offer a complete axiomatization for the epistemic logic based on this language, and show that our language has a notion of assignment normal form. Finally, we discuss the behavior of a kind of intensional distributed knowledge that can be naturally expressed in our language.
\end{abstract}

\section{Introduction}\label{sec.intro}
\emph{Simplicial Complexes} have been a useful technical tool to study distributed computing since the 1990s,\footnote{
Earliest works on distributed computing using simplicial complexes include \cite{impdist,loui1987memory,BIRAN1990420}, etc.
One may refer to \cite{herlihy2013distributed} for more details on this topic.
} 
but it was only recently that logicians also started using them to study multi-agent epistemic logic \cite{VANDITMARSCH2021100662,GOUBAULT2021104597}. 
Roughly, a simplicial complex is a set of vertices colored by names of agents,
equipped with a family of downward-closed subsets of vertices called \emph{faces} or \emph{simplexes}.
Intuitively, the colored vertices represent the local epistemic states of the agents,
while the maximal faces, called \emph{facets}, represent global epistemic states,
which are epistemically possible for the agents they contain.

A simplicial complex is called \emph{impure} if the vertices in different facets can be colored by different sets of agents.
Intuitively, this means that an agent can be \emph{alive} in some global states, while \emph{dead} in the others.
Though this is a rather natural idea,
it turns out to be quite tricky to define a reasonable semantics for the ordinary propositional modal language (with agent-labeled propositional letters and modal operators)
on impure simplicial complexes.
This is mainly because it is sometimes difficult to decide the truth value of formulas involving possibly dead agents.
For example, consider the sentence $\K_a \K_b p_c$,
which intuitively says that $a$ knows that $b$ knows that $c$ has property $p$.
Should it be true at facet $F$ in the following simplicial complex $\SC$?
\\[-1em]
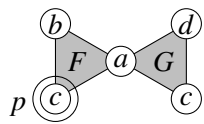
\begin{center}
\begin{tikzpicture}
\tikzstyle{vertex}=[shape=circle,draw,minimum size=4mm,inner sep=0pt,fill=white]

\node[vertex] (c1) at (0, 0) {$c$};
\node[vertex] (b) at (0, 1) {$b$};
\node[vertex] (a) at (0.87, 0.5) {$a$};
\node[vertex] (c2) at (1.74, 0) {$c$};
\node[vertex] (d) at (1.74, 1) {$d$};
\node[] (F) at (0.3 ,0.5) {$F$};
\node[] (G) at (1.44, 0.5) {$G$};
\node[shape=circle,draw,minimum size=6mm,inner sep=0pt] (pext) at (0, 0) {};
\node[] (p) at (-0.5, -0.1) {$p$};

\begin{scope}[on background layer]
\path [fill=lightgray,draw] (a.center) to (b.center) to (c1.center) to (a.center);
\path [fill=lightgray,draw] (a.center) to (d.center) to (c2.center) to (a.center);
\end{scope}
\end{tikzpicture}
\end{center}
~\\[-1em]
Our intuition here is not very clear:
since $a$ can imagine the global state $G$ where $b$ is dead,
we need to decide the truth value of $\K_b p_c$ at $G$;
but conceptually, it is hard to tell what a dead agent should know.

Hitherto, there are two kinds of impure simplicial semantics:
one is the two-valued semantics proposed in \cite{goubault2022simplicialmodelkb4nepistemic},
in which dead agents also have concrete properties --- in particular,
a dead agent always ``knows'' everything, including the contradiction;
the other is the three-valued semantics proposed in \cite{10.1007/978-3-030-88853-4_3,10.1093/logcom/exae055,SCBIS2024}, 
in which formulas involving dead agents might be neither true nor false, but \emph{undefined}.\footnote{
The relation between this two kinds of semantic is studied in \cite{SCTV2023}.
}It is worth noting that according to both semantics,
in the above example,
$\K_a \K_b p_c$ is \emph{true} at $F$.

None of the two semantics seems completely satisfactory:
in the two-valued semantics, we are forced to assign true or false to arbitrary expressions involving possibly dead agents,
which might be conceptually dubious,
\footnote{In \cite{goubault2022simplicialmodelkb4nepistemic}, the authors indeed offer an conceptual explanation for assigning non-trivial truth values of propositional letters to dead agents:
as the authors explain, propositional letters might describe what the dead agents did \emph{before they died}. 
However, it seems that such an explanation only works when we interpret \emph{all} propositional letters in this rather specific way.
}
and it sounds absurd to suggest that a dead agent \emph{knows} the contradiction;
on the other hand, the three-valued semantics relies on a rather complicated notion of definability, 
making it a bit difficult to study its behavior on the technical level,
and the corresponding axiomatization is also a non-classical non-normal modal logic \cite{10.1007/978-3-030-88853-4_3,10.1093/logcom/exae055}.
It is also interesting to note that neither semantics validates the standard $\mathtt{T}$-axiom for knowledge:
in the two-valued semantics, $\K_a \phi \to \phi$ can simply be \emph{false} when $a$ is dead,
while in the three-valued semantics,
the truth of $\K_a \phi$ only implies that $\phi$ is true \emph{if defined}.

In this paper, we propose an alternative way to do epistemic logic on impure simplicial complexes,
which seems more well-behaved both conceptually and technically.
Our strategy is to \emph{avoid} talking about dead agents in any non-trivial way.
For example, let us reconsider the example $\K_a \K_b p_c$ --- $a$ knows that $b$ knows that $c$ satisfies $p$.
For simplicity, we assume that $a$ is a living agent (as in the above model),
so that it makes perfect sense to talk about $a$'s epistemic states.
Then, while the truth value of $\K_a \K_b p_c$ is difficult to tell,
the following four ``clarifications'' of it all have obvious truth values:\\
\begin{tabular}{cl}
(i) & $a$ knows that (\emph{if $b$ is alive, then} $b$ knows that (\emph{if $c$ is alive, then} $c$ satisfies $p$)).\\
(ii) & $a$ knows that (\emph{if $b$ is alive, then} $b$ knows that (\emph{$c$ is alive and} $c$ satisfies $p$)).\\
(iii) & $a$ knows that (\emph{$b$ is alive and} $b$ knows that (\emph{if $c$ is alive, then} $c$ satisfies $p$)).\\
(iv) & $a$ knows that (\emph{$b$ is alive and} $b$ knows that (\emph{$c$ is alive and} $c$ satisfies $p$)).
\end{tabular}\\
Unlike the original expression, these clarifications involve only properties and epistemic states of \emph{living} agents, and it is quite clear that at $\SC, F$, (i), (ii) are true, while (iii), (iv) are false.
Then, instead of the ordinary propositional modal language which allows dubious expressions like $\K_a \K_b p_c$,
if we can work in a more delicate language, in which we are only allowed to talk about properties and epistemic states of living agents, then we should be able to define a fairly intuitive semantics without entering any serious philosophical controversy concerning the states of dead agents,
since philosophically controversial expressions would be excluded \emph{syntactically} from our language.

This is our motivation to introduce a special kind of \emph{term-modal language}, which is a \emph{first-order modal language} with unary predicates (viewed as labeled propositional letters), 
modal operators labeled by variables,
and \emph{assignment operators} of the form $[x \assn a] \phi$ for each agent $a$.
Intuitively, $[x \assn a] \phi$ says that \emph{if $a$ is alive}, then after assigning $a$ to $x$, $\phi$ is true; its dual, $\lr{x \assn a} \phi$, says that \emph{$a$ is alive and} after assigning $a$ to $x$, $\phi$ is true.
In this language, unary predicates and modal operators are labeled only by \emph{variables}, which refer to living agents only; on the other hand, 
names of agents appear only in assignment operators.
Thus, whenever we want to talk about an agent directly,
we are forced to use the corresponding assignment operators,
which clarify the existential state of the agent in question before we start talking about any of their concrete properties and epistemic states.
As a result, in such a language, 
sentences like (i) to (iv) above can still be expressed:
for example, assuming the existence of $a$, (ii) can be formalized as $\lr{x \assn a} \K_x [y \assn b] \K_y \lr{z \assn c} p_z$, while (iii) can be formalized as $\lr{x \assn a} \K_x \lr{y \assn b} \K_y [z \assn c] p_z$;
in contrast, dubious expressions like $\K_a \K_b p_c$ are no longer inexpressible.

It should be noted that in the literature, 
assignment operators are usually of the form $[x \assn t] \phi$, where $t$ is a \emph{term}.
This kind of assignment operators are firstly introduced in \cite{ea30231294d84616b42bc480e77c90a7} as a part of a more expressive first-order dynamic language;
more closely related to our work, a quantifier-free term-modal logic with assignment operators is studied in \cite{DBLP:journals/apal/WangWS22}.
The semantics of $[x \assn t] \phi$ is that, after assigning the value of $t$ to $x$, $\phi$ is true.
As noted in \cite{ea30231294d84616b42bc480e77c90a7}, this kind of assignment operators are also closely related to the notion of \emph{predicate abstraction} introduced in \cite{Stalnaker1968-STAAIF}:
$[x \assn t] \phi$ is equivalent with $\lr{\lambda x. \phi} (t)$.

Our assignment operators differ from the above kind of assignment operators in many important ways.
Conceptually, assignment operators of the above kind (as well as predicate abstraction) usually work a tool to capture the \emph{de re} / \emph{de dicto} distinction in first-order modal logic;
but we use assignment operators to clarify the existential states of agents. 
Technically, in $[x \assn t] \phi$, $t$ is also a term of the language, which can be used to label predicates or modal operators;
but in our operator $[x \assn a] \phi$, $a$ is \emph{not} a legitimate term, and can only appear within the assignment operator.
Also, there is no difference between $[x \assn t]\phi$ and its dual $\lr{x \assn t} \phi$,
but in our language, there is an essential difference between $[x \assn a] \phi$ and $\lr{x \assn a} \phi$: only the latter entails the existence of $a$.

In this paper, then, we will develop an epistemic logic for impure simplicial complexes based on the more delicate language with assignment operators of the form $[x \assn a] \phi$.  

The structure of the paper is as follows.
In Section \ref{sec.language}, we formally introduce the language and its semantics.
Since we are working with a kind of first-order modal language, we will also give it a natural \emph{first-order Kripke semantics}, which should help us better understand the simplicial semantics both conceptually and technically.
We will also introduce two notions of bisimulation to characterize the expressivity of the two semantics respectively.
Next, in Section \ref{sec.corr}, we show that for a special class of first-order Kripke models called \emph{local epistemic models}, the first-order Kripke semantics is equivalent to the simplicial semantics;
this resembles the correspondence between simplicial complexes and propositional Kripke models proved in \cite{GOUBAULT2021104597,10.1007/978-3-030-88853-4_3, 10.1093/logcom/exae055}.
Then, in Section \ref{sec.axiom}, we offer a complete axiomatization of the epistemic logic based on our language,
and show that our language has a notion of \emph{assignment normal form}.
Then, in Section \ref{sec.ik},
we study a kind of \emph{intensional} distributed knowledge that can be naturally expressed in our language.
Finally, in Section \ref{sec.con},
we summarize our results and consider directions for future work.

\section{Language, Semantics and Expressivity}\label{sec.language}
We first introduce the language and semantics.
Fix a non-empty finite set $\mathbf{A}$ of agents
and a countable set $\mathbf{X}$ of variables (let $\mathbf{A} \cap \mathbf{X} = \emptyset$).
Also let $\mathbf{P}$ be a countable set of unary predicates / propositional letters to be labeled.
Then, we define the language as follows.

\begin{definition}[$\LL^\sassn$]
Formulas in $\LL^\sassn$, together with the notion of \emph{free variables} of a formula, are defined recursively as follows:\footnote{
We need to define them at the same time,
since we need to use the notion of free variables when forming modal formulas.
}
\begin{center}
$\phi :: =  p_x \mid \top \mid \neg \phi \mid (\phi \wedge \phi) \mid [x \assn a] \phi \mid \K_X \alpha$\\[0.5em]

where $p \in \mathbf{P}$, $x \in \mathbf{X}$, $a \in \mathbf{A}$, $X \subseteq \mathbf{X}$ is finite ($X$ can be empty), and 
$\alpha$ is an $\LL^\sassn$-formula s.t.\ $FV(\alpha) = \emptyset$.\\[0.5em]

$FV(p_x) = \{x\}$; $FV(\top) = \emptyset$; $FV(\neg \phi) = FV(\phi)$; $FV(\phi \wedge \psi) = FV(\phi) \cup FV(\psi)$;\\
$FV([x \assn a] \phi) = FV(\phi) \setminus \{x\}$; $FV(\K_X \alpha) = X$.

\end{center}

Other Boolean connectives are defined in the usual way.

$\lr{x \assn a} \phi$ is the abbreviation for $\neg [x \assn a] \neg \phi$,
and $\DK_X \alpha$ is the abbreviation for $\neg \K_X \neg \alpha$.

We use $\phi[y/x]$ to denote the formula obtained by replacing all free occurrences of $x$ in $\phi$ with $y$,
and we say $\phi[y/x]$ is \emph{admissible} if $x$ has no free occurrence within the scope of $[y \assn a]$ for any $a \in \mathbf{A}$.

Let $\LL_0^\sassn$ denote the set of all formulas in $\LL^\sassn$ with no free variable --- 
as usual, we also call them $\LL^\sassn$-\emph{sentences}.
(In this paper, we usually use $\phi, \psi, \chi$ to denote $\LL^\sassn$-formulas in general,
and use $\alpha, \beta, \gamma$ to denote $\LL^\sassn$-sentences.)
\end{definition}

Note that the primitive notion of knowledge in our language is \emph{distributed knowledge}:
intuitively,
$\K_X \alpha$ says that the group of agents denoted by $X$ knows that $\alpha$ after putting their information together.
\footnote{
A similar kind of distributed knowledge operator labeled by a set of variables is introduced in \cite{NAUMOV_TAO_2019}.
}
As we shall see,
the distributed knowledge operators enable us to express certain model properties in a more natural way.

Next, we introduce the simplicial complex semantics and first-order Kripke semantics for this language.
The notion of simplicial models and first-order Kripke models are defined as follows:

\begin{definition}[Simplicial Model]
A \emph{simplicial model} is a $4$-tuple $\SC = (\mathcal{V}, C, \chi, \ell)$, where $\mathcal{V} \neq \emptyset$ is the set of \emph{vertices},
$C \subseteq \wp(\mathcal{V})$ is the set of \emph{faces} (or \emph{simplexes}) satisfying that 
(i) $\emptyset \notin C$, 
(ii) if $\emptyset \neq Y \subseteq X \in C$, then $Y \in C$, and 
(iii) for all $v \in \mathcal{V}$, $\{v\} \in C$;
$\chi: \mathcal{V} \to \mathbf{A}$ is the coloring function s.t.\ for all $X \in C$, $\chi_{|X}$ is an injection,
and $\ell: \mathbf{P} \to \wp(\mathcal{V})$ is the labeling function. 

Given a simplicial model $\SC = (\mathcal{V}, C, \chi, \ell)$,
let $\F(C) = \{F \in C \mid \text{for all } X \in C \setminus \{F\}, F \not \subseteq X\}$.
Elements in $\F(C)$ are called the \emph{facets} in $C$.

For convenience, we will sometimes use notations like $\mathcal{V}^\SC$, $C^\SC$, $\chi^\SC$ and $\ell^\SC$ to denote the corresponding components of a simplicial model $\SC$.
\end{definition}

\begin{definition}[First-Order Kripke Model]
In this paper, by a \emph{first-order Kripke model}, we mean a $4$-tuple $\M = (W, \delta, \{R_a\}_{a \in \mathbf{A}}, \rho)$, where
$W \neq \emptyset$ is the set of possible worlds,
$\delta: W \to \wp(\mathbf{A}) \setminus \{\emptyset\}$ assigns each $w \in W$ a local domain of agents,
for each $a \in \mathbf{A}$, $R_a \subseteq W \times W$ is the accessibility relation for $a$ s.t.\ 
for all $w \in W$, if $a \notin \delta(w)$, then $R_a(w) = \emptyset$,
and $\rho: \mathbf{P} \times W \to \wp(\mathbf{A})$ is the interpretation function that assigns each pair of $p \in \mathbf{P}$ and $w \in W$ a subset of $\delta(w)$.\footnote{
Note that according to our definition,
in every world,
only existing agents have non-trivial epistemic relations,
and the extension of a predicate only includes existing agents.
We make these assumptions here because we do not want to decide the properties and epistemic states of non-existing agents on the conceptual level,
and also because they help us establish a cleaner correspondence between simplicial models and first-order Kripke models.
}

For convenience, we will sometimes use notations like $W^\M$, $\delta^\M$, $R_a^\M$ and $\rho^\M$ to denote the corresponding components of a first-order Kripke model $\M$.
\end{definition}

Here, we do not set any further constraint on the accessibility relations in a first-order Kripke model,
so they are in a sense more \emph{free} than simplicial models.
In fact, as we shall see in Section \ref{sec.corr},
only a special class of first-order Kripke models corresponds
to simplicial models.
But our results concerning first-order Kripke models presented in this section do not rely on any such special properties.

In order to define the semantics for open formulas,
we also need the following notion of \emph{admissible assignments}.
Intuitively, such assignments ensure that an open formula only mentions \emph{existent} agents.

\begin{definition}[Admissible Assignment]
An \emph{assignment} is a function from $\mathbf{X}$ to $\mathbf{A}$.

For a simplicial model $\SC = (\mathcal{V}, C, \chi, \ell)$ and $F \in \F(C)$, $\phi \in \LL^\sassn$,
we say an assignment $\sigma$ is \emph{admissible} for $\SC, F$ and $\phi$, if $\sigma[FV(\phi)] \subseteq \chi[F]$.

Similarly, for a first-order Kripke model $\M = (W, \delta, R, \rho)$, and $w \in W$, $\phi \in \LL^\sassn$,
we say an assignment $\sigma$ is \emph{admissible} for $\M,w$ and $\phi$, if $\sigma[FV(\phi)] \subseteq \delta(w)$.
\end{definition}

Then, the two semantics are defined as follows.

\begin{definition}[The Semantics]
The satisfaction relation $\Vdash$ for simplicial models and the satisfaction relation $\vDash$ for first-order Kripke models are defined recursively as follows when the corresponding assignment is \emph{admissible} for the pointed model and formula in question
(the parts for $\top$ and Boolean connectives are omitted):\footnote{
Some remarks:
(i) For a \emph{sentence} $\alpha$, any assignment is admissible.
This guarantees that the semantics for $\K_X$ is well-defined.
(ii) In the first-order Kripke semantics for $\K_X \alpha$, we make the convention that $\bigcap_{a \in \emptyset} R_a(w) = W$,
so $\K_\emptyset$ works as the universal modality. 
}
\begin{center}
{\small
\begin{tabular}{|lcl|lcl|}
\multicolumn{3}{|l|}{Simplicial Model} & \multicolumn{3}{l}{First-Order Kripke Model}\\
\hline
$\SC, F, \sigma \Vdash p_x$ & iff & $\sigma(x) \in \chi[F \cap \ell(p)]$ & 
$\M, w, \sigma \vDash p_x$ & iff & $\sigma(x) \in \rho(p, w)$\\[0.3em]
$\SC, F, \sigma \Vdash [x \assn a] \phi$ & iff & if $a \in \chi[F]$, then & 
$\M, w, \sigma \vDash [x \assn a] \phi$ & iff & if $a \in \delta(w)$, then\\ 
& & $\SC,F,\sigma[x \mapsto a] \Vdash \phi$ & 
& & $\M,w,\sigma[x \mapsto a] \vDash \phi$\\[0.3em]
$\SC, F, \sigma \Vdash \K_X \alpha$ & iff & for all $G \in \F(C)$ s.t. &
$\M,w,\sigma \vDash \K_X \alpha$ & iff & for all $v \in W$ s.t.\\
& & $\sigma[X] \subseteq \chi[F \cap G]$, &
& & $v \in \bigcap_{a \in \sigma[X]} R_a(w)$\\
& & $\SC, G, \sigma \Vdash \alpha$ & 
& & $\M,v,\sigma \vDash \alpha$\\
\hline
\end{tabular}
}
\end{center}

Note that the truth value of an $\LL^\sassn$-sentence $\alpha$ has nothing to do with the assignment.
Thus, for $\alpha \in \LL_0^\sassn$,
we may simply write $\SC, F \Vdash \alpha$ (resp.\ $\M,w \vDash \alpha$) when $\SC, F,\sigma \Vdash \alpha$ (resp.\ $\M,w,\sigma \vDash \alpha$) for some assignment $\sigma$.
\end{definition}

It is not hard to see that the semantics of the assignment operators is exactly the one we described in Section \ref{sec.intro}.
Also note that $\lr{x \assn a} \top$ simply says that \emph{$a$ exists},
so in this sense, we can express the \emph{existential predicate} in our language.

Now, we introduce two notions of bisimulation to characterize the expressivity of the above two semantics, respectively.
Note that our notion of bisimulation for simplicial models is essentially the same as the kind of bisimulation introduced in \cite{SCBIS2024}.

\begin{definition}[Bisimulations]
Let $\SC$, $\SD$ be simplicial models,
and let $\M$, $\N$ be first-order Kripke models.
A relation $Z \subseteq \F(C^\SC) \times \F(C^\SD)$ (resp.\ $Z \subseteq W^\M \times W^\N$) is a \emph{bisimulation} between $\SC$ and $\SD$ (resp.\ $\M$ and $\N$), if the following conditions are satisfied for any $(F, G) \in Z$ (resp.\ $(w,v) \in Z$):
\begin{center}
{\small
\begin{tabular}{l|l|l}
& Simplicial Models & First-Order Kripke Models\\
\hline
& \\[-1em]
\multirow{2}{*}{(Inv)} & 
$\chi^\SC[F] = \chi^\SD[G]$, and for all $p \in \mathbf{P}$, & 
$\delta^\M(w) = \delta^\N(v)$, and for all $p \in \mathbf{P}$,\\
& $\chi^\SC[F \cap \ell^\SC(p)] = \chi^\SD[F \cap \ell^\SD(p)]$ & 
$\rho^\M(p, w) = \rho^\N(p, v)$\\[0.3em]
\hline
& \\[-1em]
& 
For all $A \subseteq \mathbf{A}$ and $F' \in \F(C^\SC)$ s.t. & 
For all $A \subseteq \mathbf{A}$ and $w' \in W^\M$ s.t. \\
(Zig) & $A \subseteq \chi^\SC[F \cap F']$, there is $G' \in \F(C^\SD)$& 
$w' \in \bigcap_{a \in A} R^\M_a(w)$, there is $v' \in W^\N$\\
& s.t.\ $A \subseteq \chi^\SD[G \cap G']$ and $(F', G') \in Z$ & 
s.t.\ $v' \in \bigcap_{a \in A} R_a^\N(v)$ and $(w', v') \in Z$\\[0.3em]
\hline
&\\[-1em]
& 
For all $A \subseteq \mathbf{A}$ and $G' \in \F(C^\SD)$ s.t. & 
For all $A \subseteq \mathbf{A}$ and $v' \in W^\N$ s.t. \\
(Zag) & $A \subseteq \chi^\SD[G \cap G']$, there is $F' \in \F(C^\SC)$& 
$v' \in \bigcap_{a \in A} R^\N_a(v)$, there is $w' \in W^\M$\\
& s.t.\ $A \subseteq \chi^\SC[F \cap F']$ and $(F', G') \in Z$ & 
s.t.\ $w' \in \bigcap_{a \in A} R_a^\M(w)$ and $(w', v') \in Z$\\
\end{tabular}
}
\end{center}

We say $\SC, F$ and $\SD, G$ (resp.\ $\M,w$ and $\N,v$) are \emph{bisimlar},
if there is a bisimulation between $\SC$ and $\SD$ (resp.\ $\M$ and $\N$) s.t.\ $(F, G) \in Z$ (resp.\ $(w, v) \in Z$).
\end{definition}

For any simplicial models $\SC, F$, $\SD,G$
and any first-order Kripke models $\M,w$, $\N,v$,
we say $\SC, F$ and $\SD, G$ (resp.\ $\M,w$ and $\N,v$) are \emph{$\LL^\sassn$-logically equivalent},
denoted by $\SC,F \equiv_{\LL^\sassn} \SD,G$ (resp.\ $\M,w \equiv_{\LL^\sassn} \N,v$),
if for any $\alpha \in \LL^\sassn_0$,
$\SC, F \Vdash \alpha$ iff $\SD, G \Vdash \alpha$ (resp.\ $\M,w \vDash \alpha$ iff $\N,v \vDash \alpha$).
Then, it is rather routine to prove the following theorem: 

\begin{theorem}
For any simplicial models $\SC,F$, $\SD,G$
and any first-order Kripke models $\M,w$, $\N,v$, 
\begin{center}
$\SC,F \bis \SD,G \implies \SC,F \equiv_{\LL^\sassn} \SD,G 
\quad \text{and} \quad
\M,w \bis \N,v \implies \M,w \equiv_{\LL^\sassn} \N,v$
\end{center}
\end{theorem}

\begin{proof}
(Sketch) If $Z$ is a bisimulation between two simplicial models $\SC$ and $\SD$, then we can show by induction that for any $\phi \in \LL^\sassn$, any $(F, G) \in Z$ and any assignment $\sigma$,
$\sigma$ is admissible for $\SC, F$ and $\phi$ iff it is admissible for $\SD, G$ and $\phi$,
and when it is admissible,
$\SC,F,\sigma \Vdash \phi$ iff $\SD, G, \sigma \Vdash \phi$.
The case for first-order Kripke models is basically the same.
\end{proof}

In order to show that both our semantics have the Hennessy-Milner property,
we define the following notions of saturation:

\begin{definition}[Saturation]
For any simplicial model $\SC = (\mathcal{V}, C, \chi, \ell)$ and any first-order Kripke model
$\M = (W, \delta, \{R_a\}_{a \in \mathbf{A}}, \rho)$,
we say $\SC$ (resp. $\M$) is \emph{saturated}, if the following hold:
\begin{center}
{\small
\begin{tabular}{l|l}
Simplicial Model & First-Order Kripke Model\\
\hline
&\\[-1em]
For any $A \subseteq \mathbf{A}$, $F \in \F(C)$ and $\Delta \subseteq \LL_0^\sassn$, & 
For any $A \subseteq \mathbf{A}$, $w \in W$ and $\Delta \subseteq \LL_0^\sassn$, \\
if $\Delta$ is finitely satisfiable in $\{G \in \F(C) \mid A \subseteq \chi[F \cap G]\}$, & 
if $\Delta$ is finitely satisfiable in $\bigcap_{a \in A} R_a(w)$,\\
then there is $G \in \F(C)$ s.t.\ $A \subseteq \chi[F \cap G]$
and $\SC, G \Vdash \Delta$ & 
then there is $v \in \bigcap_{a \in A}R_a(w)$ s.t.\ 
$\M,v \vDash \Delta$
\end{tabular}
}
\end{center}
\end{definition}

Then, we can formulate the following theorem (the proof can be found in Appendix \ref{sec.appendix}):

\begin{theorem}\label{thm.hmp}
For any saturated simplicial models $\SC,F$, $\SD,G$
and any saturated first-Kripke models $\M,w$, $\N,v$, 
\begin{center}
$\SC,F \bis \SD,G \iff \SC,F \equiv_{\LL^\sassn} \SD,G 
\quad \text{and} \quad
\M,w \bis \N,v \iff \M,w \equiv_{\LL^\sassn} \N,v$
\end{center}
\end{theorem}

\section{Equivalence of the Two Semantics}\label{sec.corr}
In \cite{GOUBAULT2021104597}, a correspondence between \emph{pure} simplicial models and a certain class of Kripke models is proved,
and such a correspondence is later generalized to \emph{impure} simplicial models in \cite{10.1007/978-3-030-88853-4_3,10.1093/logcom/exae055}.
In this section, we prove a similar result between impure simplicial models and a certain class of \emph{first-order} Kripke models.
Though we are considering first-order Kripke models rather than propositional ones,
and the simplicial complex semantics we are working with is also different from the one employed in \cite{10.1007/978-3-030-88853-4_3,10.1093/logcom/exae055} (which is three-valued), 
the idea and formal construction are still essentially the same as \cite{10.1007/978-3-030-88853-4_3,10.1093/logcom/exae055}.

Imitating the terminology in \cite{10.1007/978-3-030-88853-4_3,10.1093/logcom/exae055},
we introduce the following notion of \emph{local epistemic models}, 
which characterizes the first-order Kripke models that
correspond to simplicial models.

\begin{definition}[Local Epistemic Model]
For a first-order Kripke model $\M = (W, \delta, \{R_a\}_{a \in \mathbf{A}}, \rho)$, we say $\M$ is a \emph{local epistemic model}, if the following are satisfied:
\begin{itemize}
\item (Local $S5$) For all $a \in \mathbf{A}$, the restriction of $R_a$ to $\{w \in W \mid a \in \delta(w)\}$ is an equivalence relation;
\item (Individually Increasing Domain) For all $w, v \in W$ and $a \in \delta(w)$ s.t.\ $w R_a v$, $a \in \delta(v)$;
\item (Local Predicates) For all $p \in \mathbf{P}$, $w, v \in W$ and $a \in \rho(p, w)$ s.t.\ $w R_a v$,
$a \in \rho(p, v)$;
\item (Collectively Decreasing Domain) For all $w, v \in W$ s.t.\ $v \in \bigcap_{a \in \delta(w)} R_a(w)$, $\delta(v) \subseteq \delta(w)$.
\end{itemize}
We say $\M$ is a \emph{proper} local epistemic model, if it has the above properties plus the following one:
\begin{itemize}
\item (Properness) For all $w \in W$, $\bigcap_{a \in \delta(w)} R_a(w) = \{w\}$.
\end{itemize}
\end{definition}

All the above properties have rather intuitive meanings:
for example, 
(Individually Increasing Domain) says that every agent knows their own existence,
(Local Predicates) says that every agent knows their own atomic properties,
and (Collectively Decreasing Domain) says that at every world, all existing agents have the distributed knowledge that no other agent exists.
As we shall see in Section \ref{sec.axiom},
the first four properties also correspond to axioms in our axiomatization in a natural way.

Before showing the correspondence between local epistemic models and simplicial models,
we first show that our language cannot really distinguish local epistemic models and \emph{proper} local epistemic models,
in the sense that every local epistemic model has a logically equivalent ``properization'':

\begin{definition}
For any local epistemic model $\M = (W, \delta, \{R_a\}_{a \in \mathbf{A}}, \rho)$ and $w \in W$, for convenience, let $[w]_\delta = \bigcap_{a \in \delta(w)} R_a(w)$.
Then, the properization of $\M$ is 
$\M^{pr} = (W^{pr}, \delta^{pr}, \{R^{pr}_a\}_{a \in \mathbf{A}}, \rho^{pr})$, where $W^{pr} = \{[w]_\delta \mid w \in W\}$,
$\delta^{pr}([w]_\delta) = \{([w]_\delta, \delta(w)) \mid w \in W\}$,
$R^{pr}_a = \{([w]_\delta, [v]_\delta) \mid w R_a v\}$ for all $a \in \mathbf{A}$, 
and $\rho^{pr} = \{((p, [w]_\delta), \rho(p, w)) \mid p \in \mathbf{P}, w \in W\}$.
\end{definition}

\begin{proposition}\label{prop.pr}
For any local epistemic model $\M = (W, \delta, \{R_a\}_{a \in \mathbf{A}}, \rho)$, 
$\M^{pr}$ is a proper local epistemic model s.t.\ for any $w \in W$, $\M,w \bis \M^{pr}, [w]_\delta$
(and thus $\M,w \equiv_{\LL^\sassn} \M^{pr}, [w]_\delta$).

In particular, if $\M$ is already a proper local epistemic model,
then $\M \cong \M^{pr}$.
\end{proposition}

Now, we prove the correspondence between simplicial models and local epistemic models.

The direction from simplicial models to local epistemic models is rather straightforward:
we just take the facets as worlds,
and the colors in a facet as its local domain.
The formal construction is as follows:

\begin{definition}
For any simplicial complex $\SC = (\mathcal{V}, C, \chi, \ell)$, let $\mathtt{LEM}(\SC) = (W^\SC, \delta^\SC, \{R^{\SC}_a\}_{a \in \mathbf{A}}, \rho^\SC)$, where $W^\SC = \F(C)$, 
$\delta^\SC = \{(F, \chi[F]) \mid F \in W^\SC\}$, $R_a^\SC = \{(F, G) \in W^\SC \times W^\SC \mid a \in \chi[F \cap G]\}$ for all $a \in \mathbf{A}$, and $\rho^\SC = \{((p, F), \chi[F \cap V(p)]) \mid p \in \mathbf{P}, F \in W^\SC\}$.
\end{definition}

The converse direction is a bit more complicated,
since we need to first construct a set of vertices from the first-order Kripke model, and then construct the facets from the vertices.
The idea is to use indistinguishability cells labeled by agents as vertices.
The formal construction is as follows:

\begin{definition}
For any local epistemic model $\M = (W, \delta, \{R_a\}_{a \in \mathbf{A}}, \rho)$ and $w \in W$, $a \in \mathbf{A}$,
for convenience, let $[w]_a = R_a(w)$,
and let $F^\M_w = \{(a, [w]_a) \mid a \in \delta(w)\}$.
Then, let $\mathtt{SC}(\M) = (\mathcal{V}^\M, C^\M, \chi^\M, \ell^\M)$, where
$\mathcal{V}^\M = \{(a, [w]_a) \mid w \in W, a \in \delta(w)\}$, 
$C^\M = \{X \subseteq \mathcal{V}^\M \mid X \subseteq F^\M_w \text{ for some } w \in W\} \setminus \{\emptyset\}$,\\
$\chi^\M(a, [w]_a) = a$ for all $(a, [w]_a) \in \mathcal{V}^\M$,
and $\ell^\M(p) = \{(a, [w]_a) \mid a \in \rho(p, w)\}$ for all $p \in \mathbf{P}$.
\end{definition}

\begin{proposition}
For any simplicial model $\SC = (\mathcal{V}, C, \chi, \ell)$, $\mathtt{LEM}(\SC)$ is a \emph{proper} local epistemic model s.t.\ for any $F \in \F(C)$ and $\alpha \in \LL_0^\sassn$,
$\SC, F \Vdash \alpha$ iff $\mathtt{LEM}(\SC), F \vDash \alpha$. 

For any local epistemic model $\M = (W, \delta, \{R_a\}_{a \in \mathbf{A}}, \rho)$, $\mathtt{SC}(\M)$ is a simplicial model with $\F(C^\M) = \{F^\M_w \mid w \in W\}$ s.t.\ for any $w \in W$ and $\alpha \in \LL_0^\sassn$,
$\M,w \vDash \alpha$ iff $\mathtt{SC}(\M), F^\M_w \Vdash \alpha$.
\end{proposition}

Finally, it is also interesting to note the following properties of $\mathtt{LEM}$ and $\mathtt{SC}$:

\begin{proposition}\label{prop.iso}
For any simplicial model $\SC$, $\SC \cong \mathtt{SC}(\mathtt{LEM}(\SC))$,
and for any local epistemic model $\M$, $\M^{pr} \cong \mathtt{LEM}(\mathtt{SC}(\M))$ (so for any \emph{proper} local epistemic model $\M$, $\M \cong \mathtt{LEM}(\mathtt{SC}(\M))$).
\end{proposition}

\begin{proof}
(Sketch) Given a simplicial complex $\SC = (\mathcal{V}, C, \chi, \ell)$, we can check that the function which maps each $v \in \mathcal{V}$ to $(\chi(v), \{F \in \F(C) \mid v \in F\})$ is an isomorphism from $\SC$ to $\mathtt{SC}(\mathtt{LEM}(\SC))$.

On the other hand, given a local epistemic model $\M = (W, \delta, \{R_a\}_{a \in \mathbf{A}}, \rho)$, we can check that
if $[w]_\delta = [v]_\delta$, then $F^\M_w = F^\M_v$,
so the function which maps each $[w]_\delta$ to $F^\M_w$ is a well-defined isomorphism from $\M^{pr}$ to $\mathtt{LEM}(\mathtt{SC}(\M))$.
\end{proof}

\begin{proposition}
Let $\SC, F$, $\SD, G$ be arbitrary simplicial models,
and let $\M,w$, $\N,v$ be arbitrary local epistemic models.
Then, we have the following
\begin{center}
$\SC, F \bis \SD, G \iff \mathtt{LEM}(\SC), F \bis \mathtt{LEM}(\SD), G \quad \text{and} \quad \M,w \bis \N,v \iff \mathtt{SC}(\M), F^\M_w \bis \mathtt{SC}(\N), F^\N_v$
\end{center}
\end{proposition}

\begin{proof}
(Sketch) We can check that if $Z$ is a bisimulation between two simplicial models $\SC$ and $\SD$,
then $Z$ itself is also a bisimulation between $\mathtt{LEM}(\SC)$ and $\mathtt{LEM}(\SD)$;
conversely, if $Z$ is a bisimulation between two local epistemic models $\M$ and $\N$,
then $\{(F^\M_w, F^\N_v) \mid (w, v) \in Z\}$ is a bisimulation between $\mathtt{SC}(\M)$ and $\mathtt{SC}(\N)$.
Then, the conclusion can be proved using Proposition \ref{prop.pr} and \ref{prop.iso}.
\end{proof}

Below is a diagram of a simplicial model and a proper local epistemic model which correspond to each other via $\mathtt{LEM}$ and $\mathtt{SC}$ (up to isomorphism).
The reflexive arrows are omitted in the diagram of the local epistemic model.

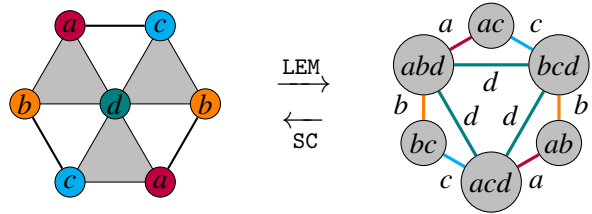
\begin{center}
\begin{tikzpicture}
\tikzstyle{avertex}=[shape=circle,draw,minimum size=4mm,inner sep=0pt,fill=purple]
\tikzstyle{bvertex}=[shape=circle,draw,minimum size=4mm,inner sep=0pt,fill=orange]
\tikzstyle{cvertex}=[shape=circle,draw,minimum size=4mm,inner sep=0pt,fill=cyan]
\tikzstyle{dvertex}=[shape=circle,draw,minimum size=4mm,inner sep=0pt,fill=teal]
\tikzstyle{world}=[shape=circle,draw,minimum size=6mm,inner sep=0pt,fill=lightgray]
\tikzstyle{bworld}=[shape=circle,draw,minimum size=8mm,inner sep=0pt,fill=lightgray]

\node[dvertex] (d) at (0, 0) {$d$};
\node[bvertex] (b1) at (-1.2, 0) {$b$};
\node[bvertex] (b2) at (1.2, 0) {$b$};
\node[avertex] (a1) at (-0.6, 1.04) {$a$};
\node[avertex] (a2) at (0.6, -1.04) {$a$};
\node[cvertex] (c1) at (-0.6, -1.04) {$c$};
\node[cvertex] (c2) at (0.6, 1.04) {$c$};

\node[world] (w1) at (5, 1.04) {$ac$};
\node[bworld] (w2) at (5.9, 0.52) {$bcd$};
\node[world] (w3) at (5.9, -0.52) {$ab$};
\node[bworld] (w4) at (5, -1.04) {$acd$};
\node[world] (w5) at (4.1, -0.52) {$bc$};
\node[bworld] (w6) at (4.1, 0.52) {$abd$};

\node[] (la1) at (4.4,1.04) {$a$};
\node[] (la2) at (5.6,-1.04) {$a$};
\node[] (lb1) at (3.8,0) {$b$};
\node[] (lb2) at (6.2,0) {$b$};
\node[] (lc1) at (4.4,-1.04) {$c$};
\node[] (lc2) at (5.6,1.04) {$c$};
\node[] (ld1) at (5,0.3) {$d$};
\node[] (ld2) at (5.26,-0.15) {$d$};
\node[] (ld2) at (4.74,-0.15) {$d$};

\node[] (SEM) at (2.5,0.4) {\large$\xrightarrow[]{\mathtt{LEM}}$};
\node[] (SC) at (2.5,-0.4) {\large$\xleftarrow[\mathtt{SC}]{}$};

\begin{scope}[on background layer]
\path [fill=lightgray,draw] (d.center) to (b1.center) to (a1.center) to (d.center);
\path [fill=lightgray,draw] (d.center) to (c2.center) to (b2.center) to (d.center);
\path [fill=lightgray,draw] (d.center) to (c1.center) to (a2.center) to (d.center);
\path[thick,draw] (a1) to (c2);
\path[thick,draw] (b1) to (c1);
\path[thick,draw] (b2) to (a2);

\path[very thick,cyan,draw] (w1) to (w2);
\path[very thick,orange,draw] (w2) to (w3);
\path[very thick,purple,draw] (w3) to (w4);
\path[very thick,cyan,draw] (w4) to (w5);
\path[very thick,orange,draw] (w5) to (w6);
\path[very thick,purple,draw] (w6) to (w1);
\path[very thick,teal,draw] (w2) to (w4);
\path[very thick,teal,draw] (w4) to (w6);
\path[very thick,teal,draw] (w6) to (w2);
\end{scope}
\end{tikzpicture}
\end{center}

\section{Axiomatization and Normal Form}\label{sec.axiom}
In this section, we offer a strongly complete axiomatization for both semantics.
Our axiomatization is different from the existing ones in interesting ways:
on the one hand, our axiomatization is based on a \emph{normal modal logic}, which distinguishes it from the one presented in \cite{10.1093/logcom/exae055};
on the other hand, 
our axiomatization is \emph{not} based on $\mathbf{KB4}$ 
(it is worth noting that we have the $\mathtt{T}$-axiom for knowledge in our axiomatization), which distinguishes it from the one presented in \cite{goubault2022simplicialmodelkb4nepistemic}.
Below is our system:\footnote{
Some notes on notations:
for a string of variables $\vec{x} = x_1, ..., x_n$
and a string of agents $\vec{a} = a_1, ..., a_n$ of the same length,
let $\{\vec{x}\}$, $\{\vec{a}\}$ denote $\{x_1, ..., x_n\}$ and $\{a_1, ..., a_n\}$ respectively,
and let $[\vec{x} \assn \vec{a}] \phi$ be the abbreviation for 
$[x_1 \assn a_1] \cdots [x_n \assn a_n] \phi$.
In particular, when $\vec{x}$ is the empty string, $[\vec{x} \assn \vec{a}] \phi$ is just $\phi$.
Also let $\K_{\vec{x}} \alpha$ (resp.\ $\DK_{\vec{x}} \alpha$) be the abbreviation for $\K_{\{\vec{x}\}} \alpha$ (resp.\ $\DK_{\{\vec{x}\}} \alpha$).
}
\begin{center}
{\small
The System $\mathbf{LEL}^\sassn$
\begin{tabular}{ll|ll}
\multicolumn{2}{l}{\textbf{Axiom Set I}} & 
\multicolumn{2}{l}{\textbf{Axiom Set II}}\\
\hline 
&&&\\[-1em]
$\mathtt{TAUT}$ & All propositional tautologies & 
$\mathtt{T}^\K$ & $\K_X \alpha \to \alpha$\\
$\mathtt{K}^\K$ & $\K_X (\alpha \to \beta) \to (\K_X \alpha \to \K_X \beta)$ & 
$\mathtt{KNI}$ & $[\vec{x} \assn \vec{a}](\neg \K_{\vec{x}} \alpha \to \K_{\vec{x}} [\vec{x} \assn \vec{a}] \neg \K_{\vec{x}} \alpha)$\\
$\mathtt{MONO}^\K$ & $\K_X \alpha \to \K_Y \alpha$ $\,$ ($X \subseteq Y$) & 
$\mathtt{EPI}$ & $[x \assn a] \K_{x} \lr{x \assn a} \top$\\
$\mathtt{K}^\sassn$ & $[x \assn a] (\phi \to \psi) \to ([x \assn a] \phi \to [x \assn a] \psi)$ & 
$\mathtt{API}$ & $[x \assn a] (p_x \to \K_{x} [x \assn a] p_x)$\\
$\mathtt{DET}^\sassn$ & $\lr{x \assn a} \phi \to [x \assn a] \phi$ & 
$\mathtt{ENI}$ & $[\vec{x} \assn \vec{a}](\bigwedge_{b \in B} [x \assn b] \bot \to \K_{\vec{x}}\bigwedge_{b \in B} [x \assn b] \bot)$\\
$\mathtt{TR}^\sassn$ & $\phi \to [x \assn a] \phi$ $\,$ ($x \notin FV(\phi)$)& 
& (where $B = \mathbf{A} \setminus \{\vec{a}\}$)\\
$\mathtt{SUB}^\sassn$ & $[y \assn a] ([x \assn a] \phi \to \phi[y/x])$ & 
\multicolumn{2}{l}{\textbf{Rules}}\\ 
\cline{3-4}
& ($\phi[y/x]$ is admissible) & 
\multicolumn{2}{l}{\multirow{3}{*}{$\displaystyle{
\mathtt{MP} \; \; \frac{\phi, \phi \to \psi}{\psi}
\quad
\mathtt{NEC}^\K \; \; \frac{\alpha}{\K_\emptyset \alpha}
\quad
\mathtt{NEC}^\sassn \; \; \frac{\phi}{[x \assn a] \phi}
}$}}\\
$\mathtt{COM}^\sassn$ & $[x \assn a] [y \assn b] \phi \to [y \assn b] [x \assn a] \phi$ $\,$ ($x \neq y$)\\
$\mathtt{UI}^\sassn$ & $\bigwedge_{a \in \mathbf{A}}[x \assn a] \phi \to \phi$\\[0.3em]
\hline 
\\[-0.5em]
\end{tabular}
}
\end{center}

Note that \textbf{Axiom Set I} contains the basic axioms for the distributed knowledge operator and the assignment operator, among which $\mathtt{SUB}^\sassn$ is probably the most interesting:
it says that in an ``environment'' where $a$ is already assigned to $y$,
assigning the same agent $a$ to $x$ \emph{semantically} is equivalent as substituting $y$ for $x$ \emph{syntatically} (which should remind the reader of the so-called \emph{Substitution Lemma} in first-order logic).
On the other hand, axioms in \textbf{Axiom Set II} characterize the special properties of local epistemic models --- 
$\mathtt{T}^\K$ and $\mathtt{KNI}$ correspond to (Local $S5$),
$\mathtt{EPI}$ corresponds to (Individually Increasing Domain),
$\mathtt{API}$ corresponds to (Local Predicates),
and $\mathtt{ENI}$ corresponds to (Collectively Decreasing Domain).
It should be noted that $\mathbf{LEL}^{\sassn}$ is \emph{not} closed under uniform substitution for labeled propositional letters:
this is because $\mathtt{API}$ only holds for atomic propositions.

Below are some useful theorems of $\mathbf{LEL}^\sassn$
(the proof can be found in Appendix \ref{sec.appendix}):

\begin{proposition}\label{prop.lelthm}
The following are $\mathbf{LEL}^\sassn$-theorems:

\begin{tabular}{ll}
$\mathtt{R}^\sassn$ & $[x \assn a] \phi \leftrightarrow [y \assn a] \phi[y/x]$ $\;$ ($y$ not in $\phi$)\\
$\mathtt{KPI}$ & $[\vec{x} \assn \vec{a}] (\K_{\vec{x}} \alpha \to \K_{\vec{x}} [\vec{x} \assn \vec{a}] \K_{\vec{x}} \alpha)$\\
$\mathtt{ANI}$ & $[x \assn a] (\neg p_x \to \K_{x} [x \assn a] \neg p_x)$\\
\end{tabular}
\end{proposition}

Using techniques of first-order modal logic,
we can prove the following completeness theorem (again, the proof can be found in Appendix \ref{sec.appendix}):

\begin{theorem}\label{thm.comp}
Under first-order Kripke semantics, $\mathbf{LEL}^\sassn$ is sound and strongly complete w.r.t.\ the class of all (proper) local epistemic models.
\end{theorem}

Then, with the help of the correspondence result we prove in Section \ref{sec.corr}, it follows that 

\begin{corollary}
Under simplicial complex semantics, $\mathbf{LEL}^\sassn$ is sound and strongly complete w.r.t.\ the class of all simplicial models.
\end{corollary}

In the rest of this section, 
we show that there is a notion of \emph{normal form} for sentences in $\mathbf{LEL}^\sassn$:
every $\LL^\sassn$-sentence is equivalent to a sentence in which the assignment operators only appear right before atomic propositions, $\bot$ and distributed knowledgeusu operators.
(It should be noted that such a notion of normal form is \emph{not} used in the completeness proof presented in Appendix \ref{sec.appendix}.)

To put it more formally, we define the following notion of \emph{assignment normal form}:
$$\LL^\sassn_\mathtt{ANF} \ni \gamma :: = [x \assn a] p_x \mid [x \assn a] \bot \mid \top \mid \neg \gamma \mid (\gamma \wedge \gamma) \mid [\vec{x} \assn \vec{a}] \K_{\vec{x}} \gamma$$

Then, we have the following proposition (the proof can be found in Appendix \ref{sec.appendix}):

\begin{proposition}\label{prop.anf}
For all $\alpha \in \LL^\sassn_0$,
there is $\gamma \in \LL^\sassn_\mathtt{ANF}$ s.t.\ $\vdash \alpha \leftrightarrow \gamma$.
\end{proposition}

Note that $[x \assn a] p_x$ can simply be viewed as a propositional letter labeled by $a$,
$[\vec{x} \assn \vec{a}] \K_{\vec{x}} \gamma$ can be viewed as a distributed knowledge formula labeled by the set $\{\vec{a}\}$,
and $[x \assn a] \bot$ works as a \emph{non-existence} predicate for $a$, which can be expressed by $\K_a \bot$ in the two-valued approach (or more directly, by the negation of the global atom $\neg a$ employed in some works of the three-valued approach \cite{SCTV2023, SCBIS2024}).
In this way, $\LL^\sassn_\mathtt{ANF}$ can be viewed as an ordinary propositional modal language equipped with a special version of two-valued simplicial semantics,
in which $p_a$ is automatically \emph{true} when $a$ does \emph{not} exist.\footnote{
This is somewhat like the ``opposite'' of \emph{negative free logic}, in which an atomic formula is automatically \emph{false} when it mentions non-existing object(s).
But note that as long as we can express the existential predicate in our language,
in terms of expressivity,
there is no essential difference between them.
}

This connects our approach and the two-valued approach in a natural way: 
the two-valued language can be viewed as a fragment of our language,
which has the same expressivity as the latter.
Under this perspective, formulas like $\K_{\vec{a}} \alpha$ in the two-valued language become abbreviations for $[\vec{x} \assn \vec{a}] \K_{\vec{x}} \alpha$,
so for example, $\K_a \alpha$ actually says that ``\emph{if $a$ exists}, then the $a$ knows that $\alpha$'',
rather than simply ``$a$ knows that $\alpha$''.
Thus, some peculiarities of two-valued approach are no longer peculiar:
for example, for any $\phi$, we have $\K_a \phi$ when $a$ does not exist,
not because a dead agent ``knows'' the everything, but because $[x \assn a] \bot$ implies $[x \assn a] \K_x \phi$;
and the $\mathtt{T}$-axiom for $\K_a$ fails because $[x \assn a] \K_x \alpha \to \alpha$ is true only under the assumption $\lr{x \assn a} \top$, but not because knowledge does not imply truth.

It is also worth noting that the connection between our approach and the two-valued one gives us the \emph{decidability} of our logic.
The satisfiability problem for $\LL^{\sassn}_\mathtt{ANF}$ is decidable, since it can be viewed as a kind of two-valued language with an essentially $\mathbf{KB4}$-logic, which is clearly decidable;
moreover, computing the normal form of an $\LL^\sassn$-sentence is also a decidable matter (as one may check in the proof of Proposition \ref{prop.anf} in Appendix \ref{sec.appendix}).
Thus, we have

\begin{theorem}
The satisfiability problem for $\LL^\sassn$-sentences is decidable.
\end{theorem}

\section{Intensional Distributed Knowledge}\label{sec.ik}
Since we have open formulas in our language,
it is natural to define groups of agents with them.
Formally, for a formula $\phi(x)$ with at most one free variable $x$ (let $\LL^\sassn_x$ denote the set of all such formulas), the group defined by $\phi(x)$ at a pointed simplicial model $\SC, F$ is $\phi(\SC, F) = \{a \in \chi^\SC[F] \mid \SC, F, \sigma[x \mapsto a] \Vdash \phi(x)\}$,
while the group it defines at a pointed first-order Kripke model $\M,w$ is $\phi(\M,w) = \{a \in \delta^\M(w) \mid \M,w,\sigma[x \mapsto a] \vDash \phi(x)\}$ (our choice of $\sigma$ does not matter here).
Note that such groups are \emph{intensional},
in the sense that the same formula $\phi(x)$ might define different groups at different facets / worlds.\footnote{
The idea of introducing intensional groups to epistemic logic dates back to \cite{10.1093/logcom/3.4.345,Grove1995NamingAI},
and intensional distributed knowledge has been studied in \cite{B_lkov__2021,B_lkov__2023}, for example. Unlike these works, however, the intensional groups here are denoted by \emph{open formulas} in our language, rather than a distinct set of \emph{group names}.
}

Then, it seems interesting to consider a kind of \emph{intensional distributed knowledge} based on such intensional groups.
Formally, for each $\phi \in \LL^\sassn_x$, the intended truth conditions for $\K_{\phi} \alpha$ are the following:

\begin{center}
{\small
\begin{tabular}{|lcl|lcl|}
\hline
\\[-1em]
$\SC, F \Vdash \K_{\phi} \alpha$ & iff & for all $G \in \F^\SC(C)$ s.t. & 
$\M,w \vDash \K_{\phi} \alpha$ & iff & for all $v \in W^\M$ s.t. \\
& & $\phi(\SC, F) \subseteq \chi^\SC[F \cap G]$, &
& & $v \in \bigcap_{a \in \phi(\M,w)} R_a^\M(w)$,\\
& & $\SC, G \Vdash \alpha$ & 
& & $\M,v \vDash \alpha$\\
\hline
\end{tabular}
}
\end{center}
In our language, we can define $\K_{\phi}$ as follows:
first, for each $\phi \in \LL^\sassn_x$,
let $\phi!(A) = \bigwedge_{a \in A} \lr{x \assn a} \phi \wedge \bigwedge_{b \in \mathbf{A} \setminus A} [x \assn b] \neg \phi$.
Intuitively, $\phi!(A)$ says that $A$ is exactly the group defined by $\phi$ at the present world.
Then, $\K_\phi \alpha$ can be defined as follows:
$$\K_\phi \alpha = \bigwedge_{\{\vec{a}\} \subseteq \mathbf{A}} (\phi!(\{\vec{a}\}) \to [\vec{x} \assn \vec{a}] \K_{\vec{x}} \alpha)$$
It is not hard to check that the truth conditions of $\K_\phi \alpha$ are exactly the ones above.

The logical behavior of $\K_\phi$ is quite interesting.
The logic of $\K_\phi$ is at least a normal modal logic with the $\mathtt{T}$-axiom $\K_\phi \alpha \to \alpha$,
but both $\K_\phi \alpha \to \K_\phi \K_\phi \alpha$ and $\neg \K_\phi \alpha \to \K_\phi \neg \K_\phi \alpha$ are \emph{invalid} in general,
even though for any rigid group $A = \{\vec{a}\}$,
if we define $\K_A \alpha$ as $[\vec{x} \assn \vec{a}]\K_{\vec{x}} \alpha$,
then both $\K_A \alpha \to \K_A \K_A \alpha$ and $\neg \K_A \alpha \to \K_A \neg \K_A \alpha$ are valid (by $\mathtt{KPI}$, $\mathtt{KNI}$ and $\mathtt{EPI}$): 
essentially, this is because an intensional group might not know its \emph{extension}.
Nevertheless, \emph{positive introspection} is still relatively easy to achieve for intensional groups:
the following lemma tells us that when $\phi(x)$ expresses a property that can be positively introspected by the agent possessing it, 
$\K_\phi$ satisfies the principle of positive introspection:

\begin{lemma}\label{lem.pi}
For any $\phi \in \LL_x^\sassn$,
if $\vdash [x \assn a] (\phi \to \K_x [x \assn a] \phi)$ for all $a \in \mathbf{A}$,
then $\vdash \K_\phi \alpha \to \K_\phi \K_\phi \alpha$.
\end{lemma}

\begin{proof}
(Sketch) If $\vdash [x \assn a] (\phi \to \K_x [x \assn a] \phi)$ for all $a \in \mathbf{A}$,
then for any $\{\vec{a}\} \subseteq \mathbf{A}$,
we can check $\vdash \phi!(\{\vec{a}\}) \to [\vec{x} \assn \vec{a}] \K_{\vec{x}} \bigvee_{\{\vec{a}\} \subseteq B \subseteq \mathbf{A}} \phi!(B)$.
Moreover, $\vdash [\vec{x} \assn \vec{a}] \K_{\vec{x}} \alpha \to [\vec{x} \assn \vec{a}] \K_{\vec{x}} [\vec{x} \assn \vec{a}] \K_{\vec{x}} \alpha$,
and for any $\{\vec{a}\} \subseteq \{\vec{b}\} \subseteq \mathbf{A}$,
$\vdash [\vec{x} \assn \vec{a}] \K_{\vec{x}} \alpha \to [\vec{y} \assn \vec{b}] \K_{\vec{y}} \alpha$.
Thus, $\vdash \K_\phi \alpha \to (\phi!(\{\vec{a}\}) \to [\vec{x} \assn \vec{a}] \K_{\vec{x}} \K_\phi \alpha)$ for all $\{\vec{a}\} \subseteq \mathbf{A}$.
Then, by definition, $\vdash \K_\phi \alpha \to \K_\phi \K_\phi \alpha$.
\end{proof}

Then, we can show that for a natural kind of intensional groups, the corresponding notion of intensional distributed knowledge indeed satisfies positive introspection:

\begin{proposition}
Fix a variable $x$,
and let $\chi$ be a formula generated by the following scheme:\footnote{
It is interesting to note the truth value of a formula $\chi$ generated this way can be defined on \emph{vertices} instead of \emph{facets}:
notice that for any $\SC = (\mathcal{V}, C, \chi, \ell)$ and $v \in \mathcal{V}$, $F, G \in \F(C)$, if $v \in F \cap G$ and $\sigma(x) = \chi(v)$,
then $\SC, F, \sigma \Vdash \chi$ iff $\SC, G, \sigma \Vdash \chi$;
thus, we are justified to simply write $\SC, v \Vdash \chi$ instead.
}
\begin{center}
$\chi :: = p_x \mid \K_x \alpha \mid \neg \chi \mid (\chi \wedge \chi)$\\[0.5em]
where $\alpha \in \LL_0^\sassn$
\end{center}
Then, for any $\beta \in \LL_0^\sassn$, $\vdash \K_\chi \beta \to \K_\chi \K_\chi \beta$.
\end{proposition}

\begin{proof}
(Sketch) By Lemma \ref{lem.pi}, it suffices to check that for any such $\chi$, $\vdash [x \assn a] (\chi \to \K_x [x \assn a] \chi)$ for any $a \in \mathbf{A}$.
In order to do so, notice that every such $\chi$ can be written in \emph{disjunctive normal form},
and when $\chi$ is $p_x$, $\neg p_x$, $\K_x \alpha$ or $\neg \K_x \alpha$,
by $\mathtt{API}$, $\mathtt{ANI}$, $\mathtt{KPI}$ and $\mathtt{KNI}$, $\vdash [x \assn a] (\chi \to \K_x [x \assn a] \chi)$ for any $a \in \mathbf{A}$.
\end{proof}

On the other hand, \emph{negative introspection} is much more difficult to achieve.
Some rather boring intensional groups indeed have distributed knowledge with negative introspection,
e.g.\ the group defined by $\bot$ and the group defined by $\top$.
The former is too weak, since $\K_\bot \alpha$ is just $\K_\emptyset \alpha$;
and the latter is too strong:
$\K_\top \alpha$ expresses the distributed knowledge of \emph{all existing agents}, so $\alpha \to \K_\top \alpha$ is actually valid for any sentence $\alpha$.\footnote{
This is easy to check by the semantics, especially by the simplicial complex semantics.
However, deducing $\alpha \to \K_\top \alpha$ in $\mathbf{LEL}^\sassn$ is not easy.
We need to make use of the fact that every formula in $\LL^\sassn$ has an assignment normal form,
and then show that for each $\gamma$ in normal form,
we indeed have $\vdash \gamma \to \K_\top \gamma$.
}
More generally, for any sentence $\alpha$, the distributed knowledge of the intensional group defined by $\K_\emptyset \alpha$ also have negative introspection,
since in each model,
it either always defines the empty group,
or always defines the group of all existing agents.
Unfortunately, we have not yet found any interesting $\phi \in \LL_x^\sassn$ that validates $\neg \K_\phi \alpha \to \K_\phi \neg \K_\phi \alpha$.

\section{Conclusion}\label{sec.con}
In this paper, we introduced a term-modal language with assignment operators to study epistemic logic on impure simplicial complexes.
We characterized the expressivity of the simplicial semantics and first-order Kripke semantics of this language,
showed the correspondence between simplicial complexes and local epistemic models,
offered a complete axiomatization for the epistemic logic based on our language, 
showed that the our language has assignment normal form,
and also studied the behavior of the intensional distributed knowledge operators, which can be defined in our language.

There are a number of interesting directions for future work.
For example, it is interesting to further study the connection between our simplicial semantics and the existing approaches,
especially the \emph{three-valued} one.
Though sentences in our language are two-valued,
there might be some natural way to generate a three-valued semantics from our semantics.
If this can be done,
then it should also be interesting to see whether such a three-valued semantics coincides with the existing ones.

We can also use our language to do epistemic logic on more general kinds of structures.
On the one hand, properties of local epistemic models clearly can be relaxed;
on the other hand, it is also interesting to consider
\emph{face-semantics} rather than \emph{facet-semantics}, like the one studied in \cite{SCTV2023},
or more general simplicial structures based on \emph{semi-simplicial sets}, like the one studied in \cite{CACHIN2025114902}.

The notion of intensional distributed knowledge also seems interesting.
Our results concerning such a notion of knowledge presented in this paper are rather primitive, 
and understanding the behavior of such intensional distributed knowledge operators might help us better understand the structure of the definable groups in such models.

Finally, it should also be interesting to generalize the kind of assignment operators we introduced in this paper to study varying-domain first-order modal logic in general.
This might provide us with an interesting first-order modal logic alternative to free logic-based varying-domain first-order modal logic.

\section*{Acknowledgement}
I would like to thank Yanjing Wang and Hans van Ditmarsch for their advice on the early versions of this paper.
I would also like to thank three anonymous reviewers from TARK 2025 for their insightful comments on this paper.

\bibliographystyle{eptcs}
\bibliography{generic}

\appendix

\section{Appendix}\label{sec.appendix}

\subsection{Proof of Theorem \ref{thm.hmp}}
\begin{proof}
We only consider the case for simplicial models.

Let $\SC$ and $\SD$ be arbitrary saturated simplicial models,
and let $Z = \{(F, G) \in \F(S^\SC) \times \F(C^\SD) \mid \SC, F \equiv_{\LL^\sassn} \SD, G\}$.
We show that $Z$ is a bisimulation between $\SC$ and $\SD$:
Let $(F, G) \in Z$ be arbitrary.
We only check that (Inv) and (Zig) hold.
 
For (Inv): Let $A = \chi^\SC[F]$,
and consider the sentence $\alpha_A \assn \bigwedge_{a \in A} \lr{x \assn a} \top \wedge \bigwedge_{b \in \mathbf{A} \setminus A} [x \assn b] \bot$.
Clearly $\SC, F \Vdash \alpha_A$, so $\SD, G \Vdash \alpha_A$,
and thus $\chi^\SD[G] = A = \chi^\SC[F]$.
Then, for any $p \in \mathbf{P}$, let $A_p = \chi^\SC[F \cap \ell^{\SC}(p)]$,
and consider the sentence $\alpha_{A_p} \assn \bigwedge_{a \in A_p} \lr{x \assn a} p_x \wedge \bigwedge_{b \in \mathbf{A} \setminus A_p} [x \assn b] \neg p_x$.
Again, $\SC, F \Vdash \alpha_{A_p}$, so $\SD, G \Vdash \alpha_{A_p}$,
and thus $\chi^\SD[G \cap \ell^\SD(p)] = A_p = \chi^\SC[F \cap \ell^\SC(p)]$.

For (Zig): Suppose (towards a contradiction)
that there are $A \subseteq \mathbf{A}$, $F' \in \F(C^\SC)$
s.t.\ $A \subseteq \chi^\SC[F \cap F']$,
but for any $G' \in \F(C^\SD)$ s.t.\ $A \subseteq \chi^\SD[G \cap G']$, $(F', G') \notin Z$.
Then, by the definition of $Z$, for every $G' \in \F(C^\SD)$ s.t.\ $A \subseteq \chi^\SD[G \cap G']$,
there is a sentence $\alpha_{G'}$ s.t.\ 
$\SC, F' \Vdash \alpha_{G'}$ but 
$\SD, G' \not \Vdash \alpha_{G'}$.
Then, since $\SD$ is saturated, there is some sentence $\beta$
s.t.\ $\SC, F' \Vdash \beta$ but $\SD, G' \not \Vdash \alpha$
for all $G' \in \F(C^\SD)$ s.t.\ $A \subseteq \chi^\SD[G \cap G']$.
Assume that $A = \{a_1, ..., a_n\}$,
and let $X = \{x_1, ..., x_n\}$.
Then, $\SD, G \not \Vdash \lr{x_1 \assn a_1} \cdots \lr{x_n \assn a_n} \DK_X \beta$ but $\SC, F \Vdash \lr{x_1 \assn a_1} \cdots \lr{x_n \assn a_n} \DK_X \beta$, causing a contradiction.
\end{proof}

\subsection{Proof of Proposition \ref{prop.lelthm}}
\begin{proof}
$\mathtt{R}^\sassn$:
Assume that $y$ is not in $\phi$.
By $\mathtt{SUB}^\sassn$,
$\vdash [y \assn a] ([x \assn \phi] \phi \to \phi[y/x])$.
Then, by $\mathtt{K}^\sassn$,
$\vdash [y \assn a] [x \assn a] \phi \to [y \assn a] \phi[y/x]$.
Finally, by $\mathtt{TR}^\sassn$,
$\vdash [x \assn a] \phi \to [y \assn a][x \assn a] \phi$,
so $\vdash [x \assn a] \phi \to [y \assn a] \phi[y/x]$.
The converse direction is similar.

$\mathtt{KPI}$:
First, by $\mathtt{SUB}^\sassn$ and $\mathtt{COM}^\sassn$, we have
$\vdash [\vec{x} \assn \vec{a}] (\K_{\vec{x}} \alpha \to \lr{\vec{x} \assn \vec{a}} \K_{\vec{x}} \alpha)$.
Then, by $\mathtt{T}^\K$, we have
$\vdash [\vec{x} \assn \vec{a}] (\K_{\vec{x}} \alpha \to \DK_{\vec{x}} \lr{\vec{x} \assn \vec{a}} \K_{\vec{x}} \alpha)$.
Hence, by $\mathtt{KNI}$,
$\vdash [\vec{x} \assn \vec{a}] (\K_{\vec{x}} \alpha \to \K_{\vec{x}} [\vec{x} \assn \vec{a}] \DK_{\vec{x}} \lr{\vec{x} \assn \vec{a}} \K_{\vec{x}} \alpha)$.
Finally, by $\mathtt{KNI}$ and $\mathtt{K}^\sassn$, 
$\vdash [\vec{x} \assn \vec{a}] \DK_{\vec{x}} \lr{\vec{x} \assn \vec{a}} \K_{\vec{x}} \alpha \to [\vec{x}\assn\vec{a}] \K_{\vec{x}} \alpha$, so
$\vdash [\vec{x} \assn \vec{a}] (\K_{\vec{x}} \alpha \to \K_{\vec{x}} [\vec{x} \assn \vec{a}] \K_{\vec{x}} \alpha)$.

$\mathtt{ANI}$:
By $\mathtt{API}$ and $\mathtt{K}^{\sassn}$, 
$\vdash \lr{x \assn a} p_x \to \lr{x \assn a} \K_{x} p_x$,
so $\vdash [x \assn a] (\DK_{x} \lr{x \assn a} p_x \to \DK_{x} \lr{x \assn a} \K_{x} p_x)$.
By $\mathtt{KNI}$ and $\mathtt{T}^\K$, $\vdash [x \assn a] (\DK_{x} \lr{x \assn a} \K_{x} p_x \to p_x)$.
Thus, we have $\vdash [x \assn a](\DK_{x} \lr{x \assn a} p_x \to p_x)$,
so $\vdash [x \assn a](\neg p_x \to \K_{x} [x \assn a] \neg p_x)$ by contraposition.
\end{proof}

\subsection{Proof of Theorem \ref{thm.comp}}
Now, we offer a proof for Theorem \ref{thm.comp}.
The strategy is to first construct a kind of canonical \emph{quasi} model, and then use tree-unraveling to turn it into a real model.

As preparation, we first add all agents to our language as new variables:
let $\LL_\mathbf{A}^\sassn$ be the language based on the variable set $\mathbf{X} \cup \mathbf{A}$.
In the following constructions, the agents will work as the canonical names of themselves.
Then, in order to construct the canonical model,
we introduce following notions:

\begin{definition}
A set $\Delta$ of $\LL^{\sassn}_\mathbf{A}$-formulas is an \emph{$FV$-maximal consistent set} (\emph{$FV$-MCS} for short),
if $\Delta$ is consistent, and for all $\phi \in \LL^{\sassn}_\mathbf{A}$ s.t.\ $FV(\phi) \subseteq FV(\Delta)$,
either $\phi \in \Delta$ or $\neg \phi \in \Delta$.

For an $FV$-MCS $\Delta$ of $\LL^{\sassn}_\mathbf{A}$-formulas and an assignment $\sigma: FV(\Delta) \to \mathbf{A}$,
we say $(\Delta, \sigma)$ is \emph{coherent},
if 
(i) for all $a \in \mathbf{A} \cap FV(\Delta)$, $\sigma(a) = a$,
(ii) for all $x \in FV(\Delta)$, $\lr{x \assn \sigma(x)} \Delta$ is consistent,\footnote{
By ``$\lr{x \assn a} \Delta$ is consistent'',
we mean that for any finite $\Delta_0 \subseteq \Delta$,
$\lr{x \assn a} \bigwedge \Delta_0$ is consistent.
} and (iii) if $\lr{x \assn a} \top \in \Delta$, then $a \in FV(\Delta)$.
\end{definition}

We can show that a coherent pair $(\Delta, \sigma)$ has the following good properties:

\begin{lemma}\label{lem.coh}
For any $FV$-MCS $\Delta$ and assignment $\sigma: FV(\Delta) \to \mathbf{A}$ s.t.\ $(\Delta, \sigma)$ is coherent, 
\begin{itemize}
\item[(i)] If $\phi \in \Delta$ and $x \in FV(\Delta)$, then $\lr{x \assn \sigma(x)} \phi \in \Delta$;
\item[(ii)] If $[x \assn a] \phi \in \Delta$, $a \in FV(\Delta)$ and $\phi[a/x]$ is admissible, then $\phi[a/x] \in \Delta$.
\end{itemize}
\end{lemma}

\begin{proof}
Assume that $(\Delta, \sigma)$ is coherent.

(i): Assume that $\phi \in \Delta$ and $x \in FV(\Delta)$,
and suppose (towards a contradiction) that $\lr{x \assn \sigma(x)} \phi \notin \Delta$.
Note that $FV(\lr{x \assn \sigma(x)} \phi) \subseteq FV(\phi) \subseteq \Delta$,
so $\neg \lr{x \assn \sigma(x)} \phi \in \Delta$.
Then, $\lr{x \assn \sigma(x)} (\phi \wedge \neg \lr{x \assn \sigma(x)} \phi)$ is consistent by coherence,
contradicting $\mathtt{SUB}^{\sassn}$.

For (ii):
Assume that $[x \assn a] \phi \in \Delta$, $a \in FV(\Delta)$ and $\phi[a/x]$ is admissible,
and suppose (towards a contradiction) that $\phi[a/x] \notin \Delta$.
Since $a \in FV(\Delta)$ and $\phi[a/x] \notin \Delta$, $\neg \phi[a/x] \in \Delta$.
Then, by coherence, $\lr{a \assn a} ([x \assn a] \phi \wedge \neg \phi[a/x])$ is consistent,
contradicting $\mathtt{SUB}^{\sassn}$. 
\end{proof}

Then, we show that every consistent set of $\LL^\sassn_\mathbf{A}$-formulas containing no free occurrence of elements in $\mathbf{A}$ can be extended into an $FV$-MCS $\Delta$ paired with an assignment $\sigma$ s.t.\ $(\Delta, \sigma)$ is coherent.

\begin{lemma}\label{lem.lind}
For any consistent $\Gamma \subseteq \LL^{\sassn}_\mathbf{A}$ s.t.\ $\mathbf{A} \cap FV(\Gamma) = \emptyset$,
there is an $FV$-MCS $\Delta$ and a assignment $\sigma$ s.t.\ $\Gamma \subseteq \Delta$ and $(\Delta, \sigma)$ is coherent.
\end{lemma}

\begin{proof}
Let $\Gamma \subseteq \LL^{\sassn}_\mathbf{A}$ be arbitrary,
and assume that $\Gamma$ is consistent and $\mathbf{A} \cap FV(\Gamma) = \emptyset$.
We first extend $\Gamma$ into a consistent set $\Delta_0 \subseteq \LL^{\sassn}$ s.t.\ for all $a \in \mathbf{A}$,
either $[x \assn a] \bot \in \Delta_0$ or $\lr{x \assn a} \top \in \Delta_0$.
For convenience, let $\{a \in \mathbf{A} \mid \lr{x \assn a} \top \in \Delta_0\} = \{a_0, ..., a_m\}$,
and let $\vec{a} = a_0, ..., a_m$.

Then, enumerate all variables in $FV(\Delta_0)$ as $\{x_k \mid k \in \omega\}$ (possibly with repetitions since $FV(\Delta_0)$ might be finite).
We construct a sequence $(\sigma_k)_{k \in \omega}$ s.t.\ for all $k \in \omega$,
$\sigma_k$ is a function from $\{x_0, ..., x_{k-1}\}$ to $\mathbf{A}$ and 
$\lr{\vec{a} \assn \vec{a}} \lr{x_0 \assn \sigma_k(x_0)} \cdots \lr{x_{k-1} \assn \sigma_k(x_{k-1})} \Delta_0$ is consistent.

As the starting point, Let $\sigma_0 = \emptyset$. We show that it has the properties we need.

Assume that $\lr{\vec{a} \assn \vec{a}} \bigwedge \Theta$ is inconsistent for some finite $\Theta \subseteq \Delta_0$,
i.e.\ $\vdash [\vec{a}\assn \vec{a}] \neg \bigwedge \Theta$.
Then, since $\lr{x\assn a_0} \top, ..., \lr{x \assn a_m} \top \in \Delta_0$, $\Delta_0 \vdash \lr{\vec{a} \assn \vec{a}} \neg \bigwedge \Theta$. 
However, since $a_0, ..., a_m$ are not in $FV(\Delta_0)$, by $\mathtt{TR}^{\sassn}$, 
$\vdash \lr{\vec{a} \assn \vec{a}} \neg \bigwedge \Theta \to \neg \bigwedge \Theta$.
Hence, $\Delta_0 \vdash \neg \bigwedge \Theta$,
contradicting that $\Delta_0$ is consistent and $\Theta \subseteq \Delta$.

Next, given $\sigma_k$, we show how to construct $\sigma_{k+1}$.
If $x_k \in \text{dom}(\sigma_k)$, then just let $\sigma_{k+1} = \sigma_k$.
Otherwise, we show that there is some $b \in \mathbf{A}$ s.t.\ $\sigma_k \cup \{(x_k, b)\}$ has the property we need.

Suppose not.
For simplicity, let $\vec{x} = x_0, ..., x_{k-1}$.
Then, for all $b \in \mathbf{A}$, there is some finite $\Theta_b \subseteq \Delta_0$ s.t.\ $\lr{\vec{a} \assn \vec{a}} \lr{\vec{x} \assn \sigma_{k}(\vec{x})} \lr{x_k \assn b} \bigwedge \Theta_b$ is inconsistent.
By $\mathtt{COM}^{\sassn}$, it follows that $\lr{x_k \assn b} \lr{\vec{a} \assn \vec{a}} \lr{\vec{x} \assn \sigma_{k}(\vec{x})} \bigwedge \Theta_b$ is inconsistent.
Then, let $\Theta = \bigcup_{b \in \mathbf{A}} \Theta_b$.
Clearly $\vdash \bigvee_{b \in \mathbf{A}} \lr{x_k \assn b} \lr{\vec{a} \assn \vec{a}} \lr{\vec{x} \assn \sigma_k(\vec{x})} \bigwedge \Theta \to \bot$.
Moreover, by $\mathtt{UI}^\sassn$,  
$\vdash \lr{\vec{a} \assn \vec{a}} \lr{\vec{x} \assn \sigma_k(\vec{x})} \bigwedge \Theta \to \bigvee_{b \in \mathbf{A}} \lr{x_k \assn b} \lr{\vec{a} \assn \vec{a}} \lr{\vec{x} \assn \sigma_k(\vec{x})} \bigwedge \Theta$.
Thus, $\lr{\vec{a} \assn \vec{a}}\lr{\vec{x} \assn \sigma_k(\vec{x})} \bigwedge \Theta$ is inconsistent, causing a contradiction.

Therefore, we can find some $b \in \mathbf{A}$ with the intended property.
Then, let $\sigma_{k+1} = \sigma_k \cup \{(x_k, b)\}$.

Next, let $\sigma = \{(a, a) \mid a \in \mathbf{A}_0\} \cup \bigcup_{k \in \omega} \sigma_k$.
We then construct a sequence $(\Delta_k)_{k \in \omega}$ of $\LL_\mathbf{A}^{\sassn}$-formula sets s.t.\
for all $k \in \omega$, $FV(\Delta_k) \subseteq FV(\Delta_0) \cup \mathbf{A}_0$,
and for all $\vec{x} \in FV(\Delta_0) \cup \mathbf{A}_0$, $\lr{\vec{x} \assn \sigma(\vec{x})} \Delta_k$ is consistent.
As preparation, we first enumerate all $\phi$ in $\LL_\mathbf{A}^{\sassn}$ s.t.\ $FV(\phi) \subseteq FV(\Delta_0) \cup \mathbf{A}_0$ as $\{\phi_k \mid k \in \omega\}$.
$\Delta_0$ is already constructed, and $(\Delta_k)_{k > 0}$ is constructed recursively as follows:

Given $\Delta_k$ with the intended property, we show that at least one of $\Delta_k \cup \{\phi_k\}$ and $\Delta_k \cup \{\neg \phi_k\}$ has the intended property.
Suppose not.
Then, there are $\vec{x}, \vec{x}' \in FV(\Delta_k)$ and some finite $\Theta \subseteq \Delta_k$, $\Theta' \subseteq \Delta_k$ s.t.\
$\lr{\vec{x} \assn \sigma(\vec{x})} (\bigwedge \Theta \wedge \neg \phi_k)$ and $\lr{\vec{x}' \assn \sigma(\vec{x}')} (\bigwedge \Theta' \wedge \phi_k)$ are inconsistent,
so $\vdash [\vec{x} \assn \sigma(\vec{x})] [\vec{x}' \assn \sigma(\vec{x}')] (\bigwedge \Theta \to \phi_k)$ and $\vdash [\vec{x} \assn \sigma(\vec{x})] [\vec{x}' \assn \sigma(\vec{x}')] (\bigwedge \Theta' \to \neg \phi_k)$ by $\mathtt{NEC}^{\sassn}$ and $\mathtt{COM}^\sassn$.
Thus, $\lr{\vec{x} \assn \sigma(\vec{x})}\lr{\vec{x}' \assn \sigma(\vec{x}')} \bigwedge (\Theta \cup \Theta')$ is inconsistent,
causing a contradiction.
Thus, if $\Delta_k \cup \{\phi_k\}$ has the property we need, then let $\Delta_{k+1} = \Delta_k \cup \{\phi_k\}$;
otherwise, let $\Delta_{k+1} = \Delta_k \cup \{\phi_{k+1}\}$.

Finally, let $\Delta = \bigcup_{k \in \omega} \Delta_k$.
It is easy to check that $\Delta$ is an $FV$-MCS and $(\Delta, \sigma)$ is coherent.
\end{proof}

Then, we define the canonical \emph{quasi} model as follows:

\begin{definition}[Canonical Quasi-Model]
$\M^c = (W^c, \delta^c, \{R^c_A\}_{A \subseteq \mathbf{A}}, \rho^c)$ is defined as follows:
\begin{itemize}
\item $W^c = \{(\Delta, \sigma) \mid \Delta \text{ is an } FV\text{-MCS in }\LL_\mathbf{A}^{\sassn}, \sigma: FV(\Delta) \to \mathbf{A} \text{ and } (\Delta, \sigma) \text{ is coherent}\}$;
\item For all $(\Delta, \sigma) \in W^c$, $\delta^c(\Delta, \sigma) = \sigma[FV(\Delta)] = \mathbf{A} \cap FV(\Delta)$ (by coherence);
\item For all $A \subseteq \mathbf{A}$ and $(\Delta, \sigma), (\Theta, \tau) \in \Theta$, $(\Delta, \sigma) R^c_A (\Theta, \tau)$ iff $A \subseteq FV(\Delta)$
and $\K_A \alpha \in \Delta$ implies $\alpha \in \Theta$;
\item For all $p \in \mathbf{P}$ and $(\Delta, \sigma) \in W^c$, $\rho^c(p, (\Delta, \sigma)) = \{a \mid p_a \in \Delta\}$.
\end{itemize}
\end{definition}

Note that by $\mathtt{UI}^\sassn$,
$\bigvee_{a \in \mathbf{A}} \lr{x \assn a} \top$ is an $\mathbf{LEL}^\sassn$-theorem,
so for all $(\Delta, \sigma) \in W^c$,
$\delta^c(\Delta, \sigma) \neq \emptyset$.

Then, we can prove the \emph{Existence Lemma} for $\K_X$ in the usual way:

\begin{lemma}[Existence]\label{lem.exist}
For all $(\Delta, \sigma) \in W^c$,
if $\DK_A \alpha \in \Delta$,
then there is $(\Theta, \tau) \in W^c$ s.t.\ 
$(\Delta, \sigma) R^c_A (\Theta, \tau)$ and 
$\alpha \in \Theta$.
\end{lemma}

\begin{proof}
(Sketch) For any $(\Delta, \sigma) \in W^c$ and $\DK_A \alpha \in \Delta$,
it is routine to show that $\{\alpha\} \cup \{\beta \mid \K_A \beta \in \Delta\}$ is consistent and has no free variable.
Thus, by Lemma \ref{lem.lind},
it can be extended into some $(\Theta, \tau) \in W^c$,
and clearly $(\Delta, \sigma) R_A^c (\Theta, \tau)$.
\end{proof}

Using the axioms $\mathtt{MONO}^\K$, 
$\mathtt{T}^\K$
$\mathtt{KNI}$,
$\mathtt{EPI}$,
$\mathtt{API}$
and $\mathtt{ENI}$,
we can also check the following properties of $\M^c$,
which will be useful when we construct the tree-unraveling of $\M^c$:

\begin{lemma}\label{lem.cano}
For all $(\Delta, \sigma), (\Theta, \tau) \in W^c$, $A \subseteq B \subseteq \mathbf{A}$ and $a \in \mathbf{A}$, the following hold:

\begin{itemize}
\item[(i)] If $(\Delta, \sigma) R^c_B (\Theta, \tau)$, then $(\Delta, \sigma) R^c_A (\Theta, \tau)$;
\item[(ii)] The restriction of $R^c_A$ to $\{(\Delta, \sigma) \in W^c \mid A \subseteq FV(\Delta)\}$ is an equivalence relation;
\item[(iii)] If $(\Delta, \sigma) R^c_{\{a\}} (\Theta, \tau)$, then $a \in FV(\Theta)$;
\item[(iv)] If $(\Delta, \sigma) R^c_{\{a\}} (\Theta, \tau)$ and $p_a \in \Delta$, then $p_a \in \Theta$.
\item[(v)] If  $(\Delta, \sigma) R^c_{\delta^c(\Delta, \sigma)} (\Theta, \tau)$, then $\delta^c(\Theta, \tau) \subseteq \delta^c(\Delta, \sigma)$.
\end{itemize} 
\end{lemma}

\begin{proof}
(i) is quite routine to prove.

(ii): For simplicity, let $A = \{\vec{a}\}$.
By Lemma \ref{lem.coh} and the axiom $\mathtt{KNI}$,
for any $(\Delta, \sigma) \in W^c$ s.t.\ $A \subseteq FV(\Delta)$ and any sentence $\alpha$, $\DK_A \lr{\vec{x} \assn \vec{a}} \K_{\vec{x}} \alpha \to \K_A \alpha \in \Delta$.
Thus, for any $(\Delta, \sigma), (\Delta', \sigma'), (\Delta'', \sigma'') \in \{(\Delta, \sigma) \in W^c \mid A \subseteq FV(\Delta)\}$ s.t.\ $(\Delta, \sigma) R^c_A (\Delta', \sigma')$
and $(\Delta, \sigma) R^c_A (\Delta'', \sigma'')$,
if $\K_A \alpha \in \Delta'$,
then $\lr{\vec{a} \assn \vec{a}} \K_A \alpha \in \Delta'$ (by Lemma \ref{lem.coh}),
so $\DK_A \lr{\vec{a} \assn \vec{a}} \K_A \alpha \in \Delta$,
and thus $\K_A \alpha \in \Delta$,
which implies that $\alpha \in \Delta''$.
Therefore, $(\Delta', \sigma') R^c_A (\Delta'', \sigma'')$.
Moreover, by $\mathtt{T}^\K$,
for any $(\Delta, \sigma) \in W^c$ s.t.\ $A \subseteq \Delta$,
clearly $(\Delta, \sigma) R^c_A (\Delta, \sigma)$.
Hence, the restriction of $R^c_A$ to $\{(\Delta, \sigma) \in W^c \mid A \subseteq FV(\Delta)\}$ is an equivalence relation.

(iii): If $a \in FV(\Delta)$,
then by Lemma \ref{lem.coh} and $\mathtt{EPI}$,
$\K_{\{a\}} \lr{x \assn a} \top \in \Delta$,
so for any $(\Theta, \tau) \in W^c$ s.t.\ $(\Delta, \sigma) R^c_{\{a\}} (\Theta, \tau)$,
we have $\lr{x \assn a} \top \in \Theta$,
and thus $a \in FV(\Theta)$ by coherence.

(iv): If $p_a \in \Delta$,
then $a \in FV(\Delta)$,
so by Lemma \ref{lem.coh} and $\mathtt{API}$,
$p_a \in \K_{\{a\}} [x \assn a] p_x \in \Delta$,
and thus $\K_{\{a\}} [x \assn a] p_x \in \Delta$.
Hence, for any $(\Theta, \tau) \in W^c$ s.t.\ $(\Delta, \sigma) R^c_{\{a\}} (\Theta, \tau)$,
$[x \assn a] p_x \in \Theta$;
moreover, by (iii), $a \in FV(\Theta)$,
so $p_a \in \Theta$ by Lemma \ref{lem.coh}.

(v): Let $A = \delta^c(\Delta, \sigma)$ and 
$B = \mathbf{A} \setminus \delta^c(\Delta, \sigma)$.
By coherence, for all $b \in B$, $[x \assn b] \bot \in \Delta$.
Thus, by Lemma \ref{lem.coh} and $\mathtt{ENI}$,
clearly $\K_A \bigwedge_{b \in B} [x \assn b] \bot \in \Delta$.
Hence,  for any $(\Theta, \tau) \in W^c$ s.t.\ $(\Delta, \sigma) R^c_{\delta^c(\Delta, \sigma)} (\Theta, \tau)$,
$\bigwedge_{b \in B} [x \assn b] \bot \in \Theta$,
so $B \cap FV(\Theta) = \emptyset$ by Lemma \ref{lem.coh}.
Hence, $\delta^c(\Theta, \tau) \subseteq \delta^c(\Delta, \sigma)$.
\end{proof}

Now, we introduce the canonical unraveling of the canonical quasi-model.

\begin{definition}[Canonical Unraveling]
The canonical unraveling of the canonical quasi model is \\
$\U^c = (W^{uc}, \delta^{uc}, \{R^{uc}_a\}_{a \in \mathbf{A}}, \rho^{uc})$, where 
\begin{itemize}
\item $W^{uc} \subseteq (W^c \times \wp(\mathbf{A}))^{<\omega} \times W^c$, and $(\Delta_0,\sigma_0) A_0 \cdots A_{n-1}(\Delta_n,\sigma_n) \in W^c_{uc}$ iff 
for all $k < n$, the following hold:
$A_k \subseteq \mathbf{A}$
and $(\Delta_k,\sigma_k) R^c_{A_k} (\Delta_{k+1}, \sigma_{k+1})$.

For convenience, for any $s = (\Delta_0,\sigma_0) A_0 \cdots A_{n-1}(\Delta_n,\sigma_n) \in W^{uc}$,
we use $\Delta_s$ and $\sigma_s$ to denote $\Delta_n$ and $\sigma_n$, respectively.

\item For all $s \in W^{uc}$, $\delta^{uc}(s) = \delta^c(\Delta_s, \sigma_s)$.
\item For all $s, t \in W^{uc}$, $s R^{uc}_a t$ iff $s$ and $t$ are of the following forms:
\begin{itemize}
\item $s = (\Delta_0, \sigma_0) A_0 \cdots A_{n-1} (\Delta_n, \sigma_n) B_0 (\Theta_1, \tau_1) B_1 \cdots B_{m-1} (\Theta_m, \tau_m)$
\item $t = (\Delta_0, \sigma_0) A_0 \cdots A_{n-1} (\Delta_n, \sigma_n) B'_0 (\Theta'_1, \tau'_1) B'_1 \cdots B'_{m'-1} (\Theta'_{m'}, \tau'_{m'})$
\end{itemize}
where $a \in \delta^c((\Delta_n, \sigma_n)) \cap \bigcap_{i < m} B_i \cap \bigcap_{j < m'} B'_j$ (we allow $n$, $m$ and $m'$ to be $0$);
\item For all $p \in \mathbf{P}$ and $s \in W^{uc}$, $\rho^{uc}(p, s) = \rho^c(p, (\Delta_s,\sigma_s))$.
\end{itemize}

Then, for any $(\Gamma, \pi) \in W^c$,
let $\U_{(\Gamma, \pi)}^c$ be the submodel of $\U^c$ generated from $(\varepsilon, (\Gamma, \pi))$, where $\varepsilon$ is the empty string.
\end{definition}

By the definition of the canonical unraveling and Lemma \ref{lem.exist},\ref{lem.cano}, we can easily check the following:

\begin{lemma}
For any $(\Gamma, \pi) \in W^c$, the following hold:
\begin{itemize}
\item $\U^c_{(\Gamma, \pi)}$ is a local epistemic model;
\item For all $s, t$ in $\U^c_{(\Gamma, \pi)}$, if $t \in \bigcap_{a \in A} R^{uc}_a(s)$, then $(\Delta_s, \sigma_s) R^c_A (\Delta_t, \sigma_t)$;
\item For all $s$ in $\U^c_{(\Gamma, \pi)}$, if $\DK_A \alpha \in \Delta_s$,
then there is $t \in \bigcap_{a \in A} R^{uc}_a (s)$ s.t.\ $\alpha \in \Delta_t$.
\end{itemize}
\end{lemma}

Then, we can prove the following \emph{Truth Lemma}:

\begin{lemma}[Truth]\label{lem.truth}
For any $(\Gamma, \pi) \in W^c$,
any $s$ in $\U^c_{(\Gamma, \pi)}$ and any $\phi \in \LL^{\sassn}_\mathbf{A}$ s.t.\ $FV(\phi) \subseteq FV(\Delta_s)$,
$$\U^c_{(\Gamma, \pi)}, s, \sigma_s \vDash \phi \iff \phi \in \Delta_s$$
\end{lemma}

\begin{proof}
As usual, the above Truth Lemma can be proved by induction on the structure of $\LL^\sassn_\mathbf{A}$-formulas.
We only consider the case for $p_x$ and $[x \assn a] \phi$ here, since the other cases are rather routine
(Only note that when proving the case for $\K_X \alpha$,
we need to make use of the fact that $\alpha$ is a sentence:
the truth value of a sentence is irrelevant to the assignment, so we can change the assignment to apply IH.)

For $p_x$: Let $\sigma_s(x) = a$.
First, assume that $p_x \in \Delta$.
Then, since $(\Delta_s, \sigma_s)$ is coherent, 
by Lemma \ref{lem.coh},
$\lr{x \assn a} p_x \in \Delta_s$.
Then, by $\mathtt{DET}^\sassn$, $[x \assn a] p_x \in \Delta_s$;
and we also have $\lr{x \assn a} \top \in \Delta_s$.
Note that the latter implies that $a \in FV(\Delta_s)$ by coherence.
Then, by Lemma \ref{lem.coh}, $p_a \in \Delta_s$,
so $a \in \rho^c(\Delta_s, \sigma_s) = \rho^{uc}(s)$.

Conversely, if $p_x \notin \Delta$, then $\neg p_x \in \Delta$.
Then, by the same reasoning as above, we have $\neg p_a \in \Delta_s$, so $p_a \notin \Delta$,
and thus $a \notin \rho^{uc}(s)$.

For $[x \assn a] \phi$:
By $\mathtt{R}^\sassn$,
we may assume without loss of generality that $\phi[a/x]$ is admissible.

First, assume that $[x \assn a] \phi \in \Delta_s$.
Also assume that $a \in \delta^{uc}(s) = \mathbf{A} \cap FV(\Delta_s)$.
Then, by Lemma \ref{lem.coh}, we have $\phi[a/x] \in \Delta_s$.
Hence, by IH, $\U^c_{(\Gamma, \pi)}, s, \sigma_s \vDash \phi[a/x]$,
and thus $\U^c_{(\Gamma, \pi)}, s, \sigma_s[x \mapsto a] \vDash \phi$.

Conversely, assume that $[x \assn a] \phi \notin \Delta_s$.
Then, $\lr{x \assn a} \neg \phi \in \Delta_s$,
so $[x \assn a] \neg \phi \in \Delta_s$ by $\mathtt{DET}^\sassn$,
and also $\lr{x \assn a} \top \in \Delta_s$.
By coherence, the latter implies that $a \in FV(\Delta_s)$,
so by Lemma \ref{lem.coh}, $\neg \phi[a/x] \in \Delta$,
and thus $\phi[a/x] \notin \Delta$.
Hence, by IH, $\U^c_{(\Gamma, \pi)}, s, \sigma_s \not \vDash \phi[a/x]$,
and thus $\U^c_{(\Gamma, \pi)}, s, \sigma_s[x \mapsto a] \not \vDash \phi$.
\end{proof}

Now, with the help of the above lemmas,
we can easily prove the Completeness Theorem.\\

\begin{proofcomp}
Let $\Gamma_0$ be an arbitrary consistent set of $\LL^\sassn$ formulas.
By Lemma \ref{lem.lind}, there is $(\Gamma, \pi) \in W^c$ s.t.\ $\Gamma_0 \subseteq \Gamma$.
Then, by Lemma \ref{lem.truth},
we have $\U^c_{(\Gamma, \pi)}, (\varepsilon, (\Gamma, \pi)), \pi \vDash \Gamma_0$, so $\Gamma_0$ is satisfiable.
Thus, $\mathbf{LEL}^\sassn$ is strongly complete w.r.t.\ the class of all local epistemic model.
Then, by Proposition \ref{prop.pr}, $\mathbf{LEL}^\sassn$ is also strongly complete w.r.t.\ the class of all \emph{proper} local epistemic model.
\end{proofcomp}

\subsection{Proof of Proposition \ref{prop.anf}}
In order to prove Proposition \ref{prop.anf},
first notice that we have the following:

\begin{lemma}\label{lem.elmthm}
The following are $\mathbf{LEL}^\sassn$-theorems:
\begin{center}
\begin{tabular}{ll}
$\mathtt{ELM}^\mathtt{TR}_a$ & $[x \assn a] \chi \leftrightarrow ([x \assn a] \bot \vee \chi)$ $\,$ ($x \notin FV(\chi)$)\\
$\mathtt{ELM}^\neg_{a}$ & $[x \assn a] \neg \phi \leftrightarrow ([x \assn a] \bot \vee \neg [x \assn a] \phi)$\\
$\mathtt{ELM}^\wedge_{a}$ & $[x \assn a] (\phi \wedge \psi) \leftrightarrow ([x \assn a] \phi \wedge [x \assn a] \psi)$\\
$\mathtt{ELM}^\sassn_{a}$ & $[x \assn a][y \assn b] \phi \leftrightarrow [y \assn b][x \assn a] \phi$  $\,$ ($y \neq x$)\\
\end{tabular}
\end{center}
\end{lemma}

Then, with the help of the $\mathbf{LEL}^\sassn$-theorems listed above,
we can prove Proposition \ref{prop.anf}.\\

\begin{proofanf}
For convenience, 
for an assignment operator $[x \assn a]$ and a formula $\phi$,
we say $\phi$ is in \emph{$[x \assn a]$-normal form}, if for any subformula of $\phi$ of the form $[x \assn a] \theta$, $\theta$ is one of $p_x$, $\bot$ and $[\vec{y} \assn \vec{b}] \K_Z \alpha$, where $p \in \mathbf{P}$, $\vec{y} \in Z$ and $x \in Z \setminus \{\vec{y}\}$.
Then, given $[x \assn a]$, consider the following translation:

\begin{center}
\begin{tabular}{lcl|lcl}
\multicolumn{3}{l|}{When the formula is not $[x \assn a] \chi$:} & \multicolumn{3}{l}{When the formula is $[x \assn a] \chi$ and $x \in FV(\chi)$:}\\
$\mathtt{NF}^a_x(p_y)$ & $=$ & $p_y$ & 
$\mathtt{NF}^a_x([x \assn a] p_x)$ & $=$ & $[x \assn a] p_x$\\
$\mathtt{NF}^a_x(\neg \phi)$ & $=$ & $\neg \mathtt{NF}_x^a(\phi)$ &
$\mathtt{NF}^a_x([x \assn a] \neg \phi)$ & $=$ & $[x \assn a] \bot \vee \neg \mathtt{NF}_x^a([x \assn a] \phi)$\\
$\mathtt{NF}^a_x(\phi \wedge \psi)$ & $=$ & $\mathtt{NF}_x^a(\phi) \wedge \mathtt{NF}_x^a(\psi)$ & 
$\mathtt{NF}^a_x([x \assn a](\phi \wedge \psi))$ & $=$ & $\mathtt{NF}_x^a([x \assn a] \phi) \wedge \mathtt{NF}_x^a([x \assn a] \psi)$\\
$\mathtt{NF}^a_x([y \assn b] \phi)$ & $=$ & $[y \assn b] \mathtt{NF}^a_x(\phi)$ & 
$\mathtt{NF}^a_x([x \assn a][y \assn b] \phi)$ & $=$ & $[y \assn b] \mathtt{NF}^a_x([x \assn a]\phi)$\\
$\mathtt{NF}^a_x(\K_X \alpha)$ & $=$ & $\K_X \mathtt{NF}_x^a(\alpha)$ & 
$\mathtt{NF}^a_x([x \assn a] \K_X \alpha)$ & $=$ & $[x \assn a] \K_X \mathtt{NF}_x^a(\alpha)$\\
$\mathtt{NF}^a_x(\top)$ & $=$ & $\top$ & \multicolumn{3}{l}{When the formula is $[x \assn a] \chi$ and $x \notin FV(\chi)$:}\\
& & & $\mathtt{NF}^a_x([x \assn a] \chi)$ & $=$ & $[x \assn a] \bot \vee \mathtt{NF}^a_x(\chi)$\\
\end{tabular}
\end{center}

Then, using Lemma \ref{lem.elmthm}, it is rather routine to check by induction that for any $\phi \in \LL^\sassn$,
$\mathtt{NF}_x^a(\phi)$ is in $[x \assn a]$-normal form and
$\vdash \phi \leftrightarrow \mathtt{NF}_x^a(\phi)$;
moreover, if $\phi$ is already in $[y \assn b]$-normal form,
then $\mathtt{NF}^a_x(\phi)$ is still in $[y \assn b]$-normal form.

Then, given $\alpha \in \LL_0^\sassn$, we first list all the assignment operators occurring in $\alpha$ as $[x_0 \assn a_0], [x_1 \assn a_1],$ ..., $[x_{n-1} \assn a_{n-1}]$;
then, let $\alpha_0 = \alpha$,
and given $\alpha_k$, let $\alpha_{k+1} = \mathtt{NF}^{a_k}_{x_k}({\alpha_k})$.
In this way, the formula $\alpha_n$ we obtained is an assignment normal form of $\alpha$.
\end{proofanf}
\end{document}